\documentclass[conference]{IEEEtran}
\IEEEoverridecommandlockouts
\usepackage{cite}
\usepackage{amsmath,amssymb,amsfonts}
\usepackage{algorithmic}
\usepackage{graphicx}
\usepackage{subfigure}
\usepackage{caption}
\usepackage{textcomp}
\usepackage{epstopdf}
\usepackage{xcolor}
\usepackage{mathrsfs}
\usepackage{booktabs}
\usepackage{array}
\usepackage{setspace}
\usepackage{amsthm}

\newtheorem{theorem}{Theorem}  [section]
\def\BibTeX{{\rm B\kern-.05em{\sc i\kern-.025em b}\kern-.08em
    T\kern-.1667em\lower.7ex\hbox{E}\kern-.125emX}}
\begin{document}

\title{Proof of Federated Learning: A Novel Energy-recycling Consensus Algorithm\\
}
\author{
\IEEEauthorblockN{Xidi Qu\IEEEauthorrefmark{1}, Shengling Wang(Corresponding Author)\IEEEauthorrefmark{1}, Qin Hu \IEEEauthorrefmark{2}, Xiuzhen Cheng\IEEEauthorrefmark{3}}
\IEEEauthorblockA{
\IEEEauthorrefmark{1}\textit{School of Artificial Intelligence, Beijing Normal University, Beijing, China}\\
Email:{sindychanson@mail.bnu.edu.cn, wangshengling@bnu.edu.cn} \\
\IEEEauthorrefmark{2}\textit{Department of Computer and Information Science, Indiana University-Purdue University Indianapolis, Indiana, USA} \\
Email:qinhu@iu.edu \\
\IEEEauthorrefmark{3}\textit{School of Computer Science and Technology, Shandong University, Qingdao, China} \\
Email:xzcheng@sdu.edu.cn
}}
\maketitle

\begin{abstract}
Proof of work (PoW), the most popular consensus mechanism for Blockchain, requires ridiculously large amounts of energy but without any useful outcome beyond determining accounting rights among miners. To tackle the drawback of PoW,  we propose a novel energy-recycling consensus algorithm, namely  proof of federated learning (PoFL), where the energy originally wasted to solve difficult but meaningless puzzles in PoW is reinvested to federated learning. Federated learning and pooled-ming, a trend of PoW,  have a natural fit in terms of organization structure. However, the separation between the data usufruct and ownership in Blockchain lead to data privacy leakage in model training and verification, deviating from the original intention of federal learning. To address the challenge, a reverse game-based data trading mechanism and a privacy-preserving model verification mechanism are proposed. The former  can  guard against training data leakage while the latter verifies the accuracy of a trained model with privacy preservation of the task requester's test data as well as the pool's submitted model. To the best of our knowledge, our paper is the first work to employ federal learning as the proof of work for Blockchain. Extensive simulations based on synthetic and real-world data demonstrate the effectiveness and efficiency of our proposed mechanisms.
\end{abstract}

\begin{IEEEkeywords}
Blockchain, consensus mechanism, federated learning, reverse game
\end{IEEEkeywords}

\section{Introduction}
Blockchain, disrupting the current centralized models, is heralded as
the next paradigm innovation in  digital networks, which opens a door to uncharted cyberspace with ever-increasing security, verifiability and transparency concerns. The performance of Blockchain  heavily relies   on the adopted consensus mechanisms in terms of
efficiency, consistency, robustness and scalability.  The aim of consensus mechanisms is  orchestrating the global state machine so as to agree on the order of deterministic events and
screen out invalid events. Undisputedly, the most popular consensus mechanism  is PoW, which is adopted by two mainstream Blockchain systems, namely Bitcoin and Ethereum.

PoW determines  accounting rights and rewards through the
competition among nodes ({\it miners}) to solve a hard cryptographic puzzle by brute-forcing, which is called {\it mining}, an extremely  computation-hungry process. It is reported that the total electricity consumption of Bitcoin is comparable to that of Austria annually; the electricity that a single Bitcoin  transaction expends is equal to 22.32 U.S. households powered for one day \cite{Digiconomist19}. The energy-wasting way of PoW deviates from the sustainable and environment-friendly trend for current technology development, thus diluting
its value and hindering its further application.

To tackle the drawback of PoW, researchers proposed solutions from two different perspectives: {\it energy-conservation} and {\it energy-recycling}.
Proof of stake (PoS) \cite{King12} and
voting-based consensus algorithms \cite{Nguyen18} are typically  energy-conservation approaches. They economize on energy  by cutting down the mining difficulty of rich stakeholders or their delegates. Non-democracy is an obvious side effect  of these approaches since they have a bias toward wealthy peers. Furthermore, their mining process, even with reduced difficulty,  is still considered as waste on useless computation.

Energy-recycling consensus algorithms address the energy-wasting  issue of PoW from a different angle. They recycle the energy which is originally employed  to solve cryptographic puzzles for useful tasks. For instance, the mining energy can be repurposed for finding long prime chains \cite{King13}, matrix computation \cite{Shoker17}, image segmentation \cite{Li19} and deep learning \cite{Chenli19}. The idea of energy-recycling consensus algorithms, i.e., turning the meaningless proof of work into practical tasks for completing the consensus of Blockchain, undoubtedly deepens the integration of Blockchain and other fields, expanding the application scope of Blockchain.

In this paper, we propose a novel energy-recycling consensus algorithm: {\it proof of federated learning} (PoFL), where the energy originally wasted to solve difficult but meaningless puzzles
 in PoW is reinvested to federated  learning. Federated learning \cite{Bonawitz19} is a distributed machine learning approach, with the idea of bringing code to data rather than the reverse direction.  In federated  learning, a high-quality model maintained by a central server can be learned through aggregating locally-computed updates, which are  sent by  a loose federation of participating clients. Since the local training dataset of each  client will not be sent to the central server, privacy and security risks are significantly reduced.

Besides PoFL  can inherit the advantage of   energy-recycling consensus algorithms, we propose it also because PoW and federated learning have a natural fit in terms of organization structure. Due to the huge difficulty for individual mining, {\it pooled-mining } becomes a trend of PoW, where miners join pools, gathering their computational power, to figure out the cryptographic solution  and then the pool manager of each pool allocates rewards proportionally to each miner's contribution. In other words, both pooled-mining and federated learning has a {\it clustering}   structure. Thus, when the cryptographic puzzle  in PoW is replaced with federated learning, the cluster head, namely the pool manager, can coordinate the locally-computed results updated by pool members (miners)  to form a high-quality model.

However, the fit of the organization structure does not imply  it is non-trivial to realize PoFL. In detail, the merit of federated learning lies in that all clients collaboratively train a high-quality model while keeping their training data private from others, where the usufruct  and ownership of local training data is an integration.
However, the openness of Blockchain endows  anyone with the right of  mining, i.e., the local model training in PoFL, which results in the separation between the data usufruct  and ownership. This may lead to data privacy leakage in model training and verification, deviating from the original intention of federal learning.

 Our paper aims to address the above challenge so that PoFL can meet practical demands.  To the best of our knowledge, our paper is the first work to employ federal learning as the proof of work for Blockchain, where the main contributions are summarized as follows:
\begin{itemize}
  \item A general framework of PoFL is introduced, which clarifies the interaction among all  entities involved and designs a new PoFL block structure for supporting block verification so as to realize the consensus for Blockchain.
  \item A reverse game-based data trading mechanism is proposed to leverage market power for guarding against  training data leakage. This mechanism can determine the optimal data trading probability and pricing strategy even when a pool conceals  his\footnote{We denote the pool as ``he'' and the data provider as ``she'' for differentiation in the following.} profit from  privacy disclosure.  Driven by the proposed mechanism, a pool with a high  risk of privacy disclosure has a low data trading chance and needs to pay a high purchase price,  which further incentivizes pools to train models without any data leakage.
  \item A privacy-preserving model verification mechanism is designed to verify the accuracy of a trained model while  preserving  the privacy of the task requester's test data as well as the pool's submitted model. The proposed mechanism consists of two parts: the homomorphic encryption (HE)-based label prediction and the secure two-party computation (2PC)-based label comparison.
\end{itemize}

The remaining part of the paper proceeds as follows.
In Section \ref{sec:flow}, we introduce an overview of our proposed PoFL. To cope with the training data exchange between data providers and pools, we proposed an incentive-compatible data trading mechanism based on the reverse game in Section \ref{sec:IC}. In order to calculate the accuracy of model without disclosing either the test data or the model itself, we proposed a privacy-preserving model verification mechanism  based on HE and 2PC in Section \ref{sec:2PC}. We conduct an experimental evaluation to illustrate the effectiveness and efficiency of our proposed PoFL in Section \ref{sec:experiment}, and summarize the most related work in Section \ref{sec:related}. The whole paper is concluded in Section \ref{sec:conclusion}.

\section{Framework of PoFL}
\label{sec:flow}
PoFL employs federated learning to solve realistic problems with practical value to achieve consensus in Blockchain. In our proposed PoFL framework, the problems such as image recognition and semantic analysis are published as tasks on a platform by {\it requesters}, along with the corresponding rewards as incentives for mining. Considering that the amount of reward can indirectly reflect the importance and urgency of a task, we assume that the platform will choose the task with the highest reward as the current problem which should be solved by miners as a proof-of-work to reach consensus. In the case of multiple tasks with the same highest reward, the platform will select the earliest-arrived one.

As mentioned above, pooled-mining has become a development trend in Blockchain currently, which  has a similar organization structure with federated learning. Therefore, we investigate PoFL under the   pooled-mining paradigm in this paper, whose framework is shown in Figure  \ref{fig:framework}.

According to Figure  \ref{fig:framework}, pool members (miners) train machine learning (ML) models individually based on their private data to obtain locally-computed updates, which will be aggregated by  the pool manager  so as to achieve a high-quality model. This process is named as \textit{federated mining}. After accuracy computation with the requester's test data, each pool manager packages transactions and generates a new block containing the  information needed for model verification. Once receiving blocks, full nodes will identify the winner pool through verifying model accuracy. The winner pool should send his  model to the requester,  thus obtaining the accounting right and the corresponding reward.
\begin{figure}[ht]
\centerline{
\includegraphics[width=9.5cm, height=6cm]{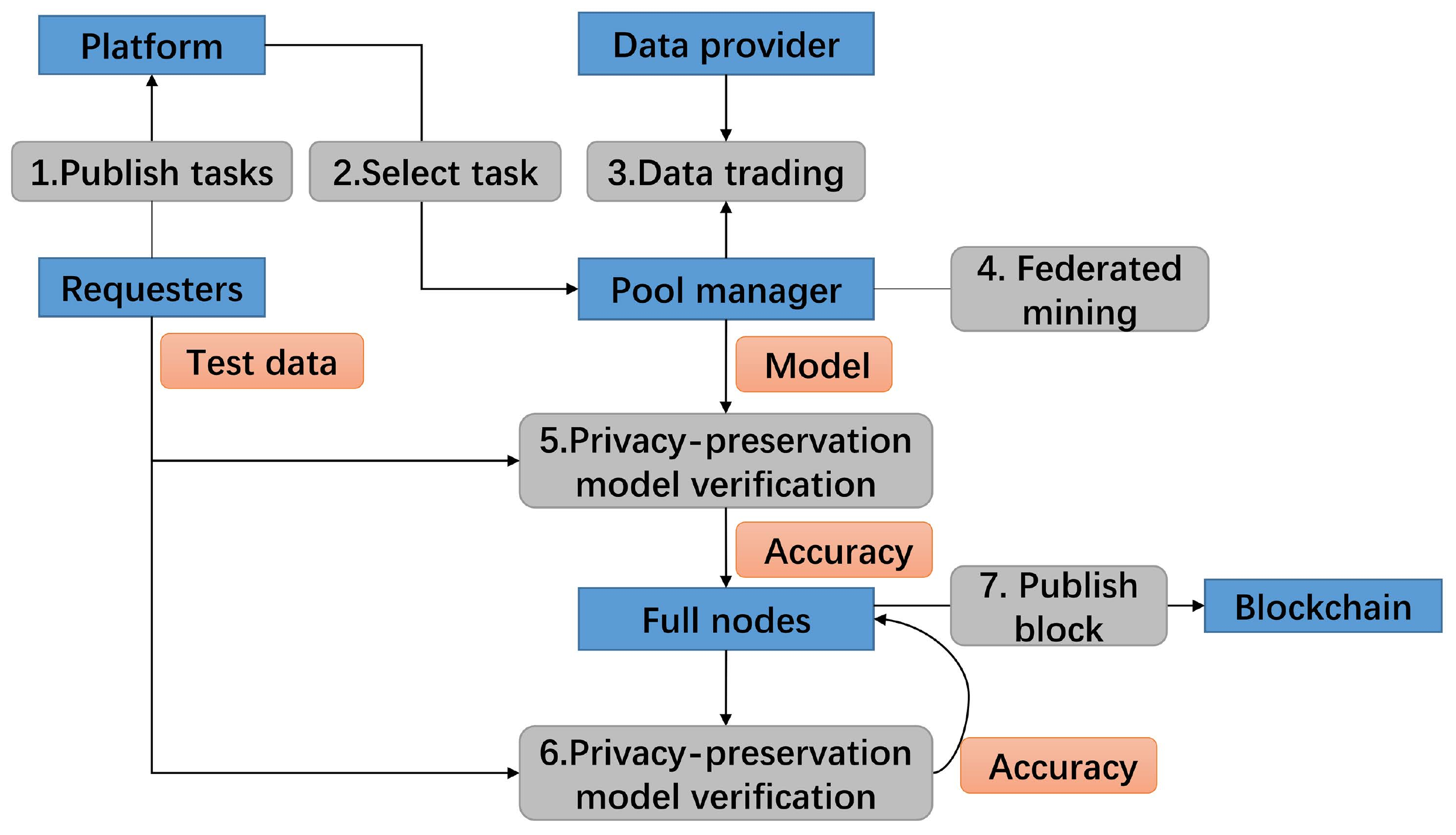}}
\caption{Framework of PoFL.}
\label{fig:framework}
\end{figure}

To implement federated learning, miners in the same pool will collectively train their models without any centralized storage, which is coordinated by the pool manager. With respect to the different storage characteristics of data, there are various types of federated learning, such as horizontal federated learning, vertical federated learning, and federated transfer learning.
Here we use horizontal federated learning as an example to illustrate the federated mining process\footnote{Other types of federated learning can also be implemented using our proposed mechanism. We omit them due to the limited length of the paper.}, in which the data of miners have the same kind of attributes belonging to different individuals.
It is shown in Figure  \ref{process} and described as follows:
\begin{figure}[ht]
\centerline{
\includegraphics[width=6cm, height=4cm]{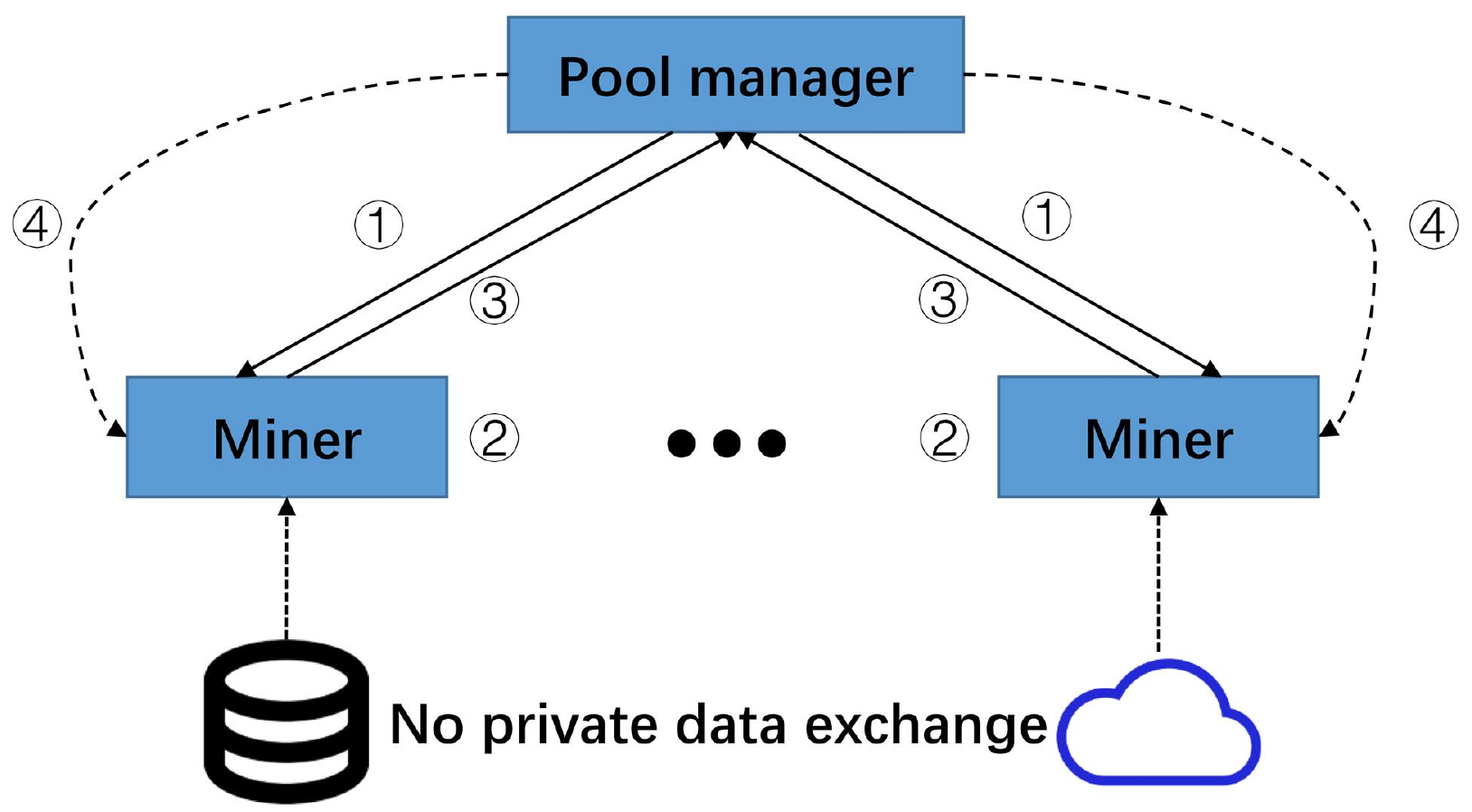}}
\caption{Federated mining process.}
\label{process}
\end{figure}

\begin{enumerate}
\item  The pool manager broadcasts an initial model as well as a public key to the pool.
\item  Each miner individually  calculates the gradient value based on the private training data and other information received from the pool manager, which is encrypted using the received public key.
\item The pool manager decrypts the locally-computed updates sent by miners, aggregating them for establishing  a shared quality-improved model.
\item According to the deadline for submitting model accuracy published by the platform, the pool manager determines whether or not the next round of training is needed. If not,  federated mining is terminated. Otherwise, the aggregated results will be sent back to each miner for implementing Steps 2 and 3 repeatedly.
\end{enumerate}

The sensitive training data could be either owned by the pool members themselves or bought from some data providers. In the latter case, the ownership and the usage right of the training data are separated, posing a potential risk of privacy leakage. To avoid this undesirable situation, we propose a reverse game-based data trading mechanism in Section \ref{sec:IC}, which leverages the power of the market so that the lower risk of the pool leaking data privacy,
the higher the probability
that the pool can buy sensitive data, and the lower the price
the pool needs to pay, thus incentivizing the pool to
behave well.

Once a pool accomplishes the training process, the pool manager will calculate the accuracy of the final model based on the test data provided by the requester. However, on one hand, the requester's test data may also be sensitive that are not suitable to directly send to the pool or any third party for  accuracy verification due to the privacy leakage concern; on the other hand,
the pool manager is not willing to publish the trained model explicitly before the end of consensus competition to avoid  being  plagiarized by other competitors\footnote{In this case, a competitor may make an opportunistic choice--not training model but  submitting a result  just slightly modified based on the contributions published  by others.}. To overcome the above challenges, we propose a privacy-preserving model verification mechanism, which is detailed in  Section \ref{sec:2PC}.  Note that this mechanism can also be employed by full nodes to verify the accuracy of all received models with the help of the requester's test data.

 However,  the current block structure does not contain any   model-related parameters, which makes the full nodes not able to  verify blocks to achieve the final consensus in Blockchain. To solve this problem, we design a new  block structure, which is shown in Figure  \ref{block}.
The proposed block header keeps some information in the existing block structure, such as the hash value of the previous block header for maintaining the chain structure, the Merkle tree root for securing transactions, the block height counting from the first block, and the tamper-resistant timestamp.
In addition, in order to facilitate other nodes verifying the accuracy of the model for each pool in the network, we include {\it task}, $V_m$ and \textit{accuracy}  in the POFL block header. To be specific, {\it task} is the current executed one by all miners, which is selected by the platform; $V_m$ is the information used for verifying model accuracy, detailed in Section \ref{sec:2PC}; and {\it accuracy} indicates the  accuracy of the ML model trained by  the pool.
\begin{figure}[ht]
\centerline{
\includegraphics[width=6.5cm, height=4cm]{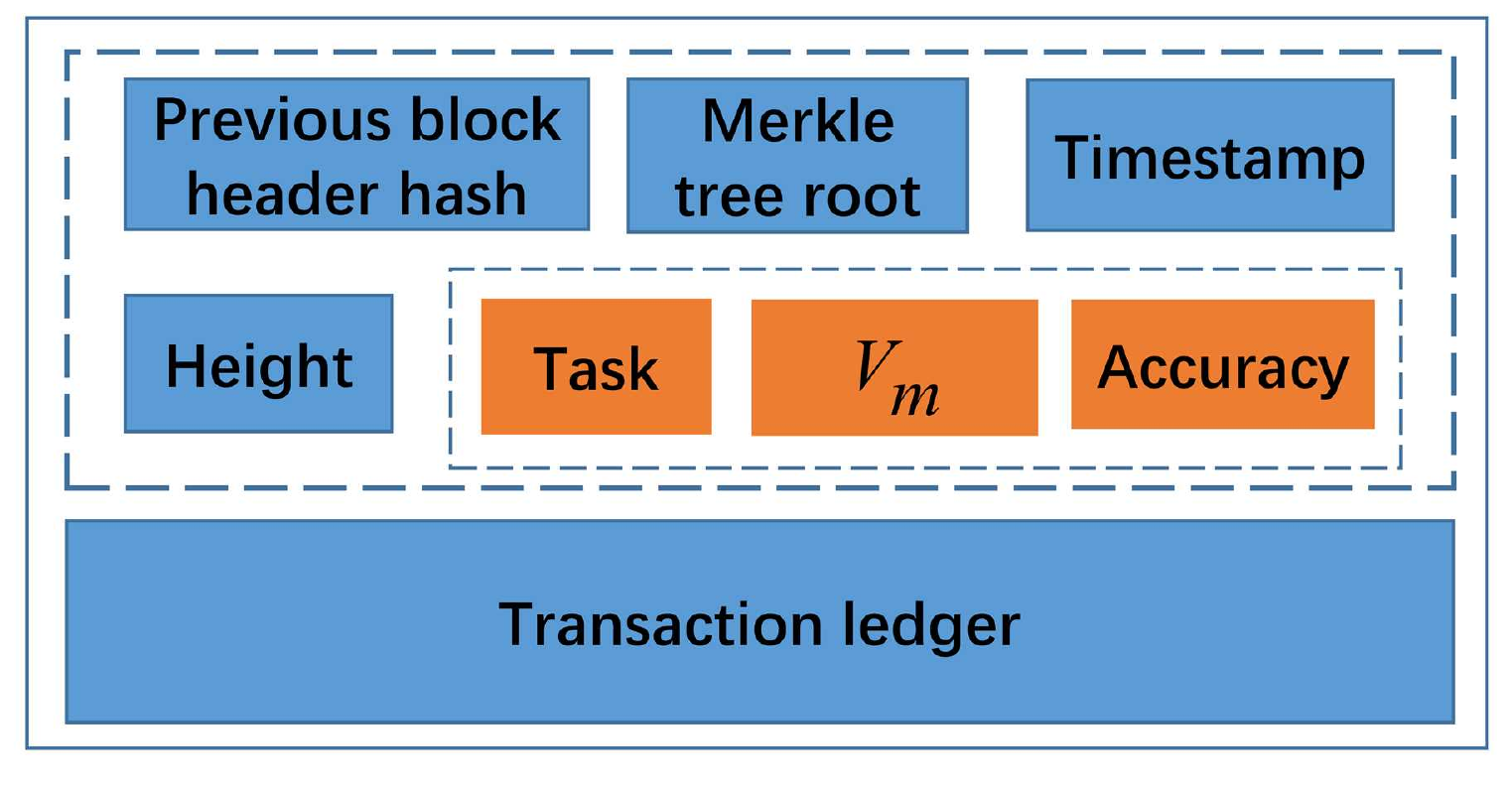}}
\caption{Block structure of PoFL.}
\label{block}
\end{figure}

\section{Reverse game-based data trading mechanism}
\label{sec:IC}
As mentioned above, when  training data are bought from some data providers,  there is a potential risk of privacy leakage due to the separation between the ownership and the usage right of the data.
An intuitive countermeasure is that data providers encrypt their sensitive data and send them to the pool for miners' training. However, this will make the training time too long, which further impacts the block generation rate and overall performance of  Blockchain.
Take the training on the encrypted MNIST dataset \cite{LeCun10} as an example. It takes 570.11 seconds to run a single instance using CryptoNets model \cite{Gilad-Bachrach16}
and nearly 57 hours to train two epochs using CryptoCNN model \cite{Xu19}. Therefore, it is not practical to employ encrypted data to train ML models for consensus in Blockchain.

To tackle the above challenge,  we propose a reverse game-based data trading mechanism, which takes advantage of   market power   to make a rational pool maximize his utility only when he trains the model without any data leakage. The proposed mechanism leverages the reverse game to describe the cooperative and conflictive relationship between the pool and the data provider, which is an efficient tool to explore the solutions for a class of private-information games. The pool also has private information in reality, such as the net profit of pooled-mining and the profit of disclosing data privacy. This information is tightly related to the probability that the pool discloses training data, which will never be told to the data provider. Leveraging on the reverse game theory, the data provider can determine the optimal trading probability and the corresponding price without knowing the private information of the pool, thus preserving the privacy  of the training data using the market as a tool.

In our mechanism, the data trading probability and its purchase price depend  on not only the private information of the pool but also his reputation.  The reputation of the pool  is calculated  according to his data privacy disclosure record. In order to establish a credible reputation mechanism, we publish the data trading records between the pool and the data provider on Blockchain. Due to the features of transparency and traceability of Blockchain, once a piece of data is disclosed, the malicious pool can be easily detected and will be accountable for this information leak with a reduced reputation. This further affects the data trading probability and the purchase price of the pool in the future.

The reverse game-based data trading mechanism selects the pool manager as the representative to conduct data trading with data providers. This design is based on two considerations: 1) It can avoid the waste of resources caused by two pool members in the same pool purchasing two copies of the same data; 2) The pool manager is able to have the knowledge of the data amount used for training in the entire pool, so that the amount of training data for each miner can be distributed as evenly as possible, which acts as the cornerstone of the average income allocation after successful mining. Once the data trading process is finished, the data provider directly sends the data to the corresponding miner, which helps avoid the communication overhead and possible privacy leakage risks caused by the pool manager's transmission\footnote{Each miner can acquire training data directly from one or more data providers which are assigned by the pool manager according to the  data trading contracts, rather than obtain all data from the pool manager who need to collect training data  from all data providers.}.

To optimize the  incomplete-information
game between the data provider and the pool, the data provider is empowered to design a game rule, which can enforce the pool  to derive the strategy based on his real private information. This requires the game rule to satisfy the incentive compatibility (IC) principle, implying that the game rule can enable the pool  to obtain a higher utility when he develops the strategy based on the real private information rather than  the fake one. In particular, the expected utility  ($U_{pool}$) of the pool within a duration $T$  is defined as:
\begin{equation}\label{up}
U_{pool}=\int_{0}^{T} p\cdot(Q+\tilde{c}(r,V)-m-D_s){\mathrm{d}t}.
\end{equation}
In \eqref{up},  $Q$ is the legally expected net profit of a pool from this data trading; $\tilde{c}(r,V)$ is the expected profit of  leaking those sensitive data,
which is closely related to the pool's reputation $r \in [0,1]$ and the value (sensitivity) of the data, denoted by $V$,  so we define it as $\tilde{c}(r,V)=\tilde{\alpha}(1-r)+\tilde{\beta}V$ with
$\tilde{\alpha}, \tilde{\beta} \geq 0$ being coefficients;
$m$ and $D_s$ are respectively the bid of the pool and the markup price proposed by the data provider, and thus the final price of the traded data is $m+D_s$; $p$ is the probability that the pool  can successfully purchase data from the data provider,
which can be obtained by
\begin{equation} \label{p}
p=\varepsilon_1 r+\varepsilon_2 \frac{D_s}{\bar{D_s}}+(1-\varepsilon_1-\varepsilon_2)\frac{m}{\bar{m}}.
\end{equation}
In \eqref{p}, $\bar{m}$ and $\bar{D_s}$ are respectively the highest values  of $m$ and $D_s$  in the recent rounds of data trading; $\varepsilon_1$ and $\varepsilon_2$ are coefficients satisfying $ \varepsilon_1, \varepsilon_2 \geq 0$ and $\varepsilon_1+\varepsilon_2 \leq 1$. The above equation indicates that the higher the reputation of the pool or the higher the final price for the sensitive data, the more willing the data provider to sell the data. Our mechanism requires that at the beginning of the data trading process, the rule for calculating the successful trading probability  shown in \eqref{p} will be informed to both parties, i.e., the data provider and the pool.

Similarly, the expected utility ($U_{provider}$) of the data provider  within a duration $T$ is defined as:
\begin{equation}
U_{provider}=\int_{0}^{T} p\cdot(m+D_s-c(r,V)){\mathrm{d}t},
\label{provider_benefit}
\end{equation}
In \eqref{provider_benefit},  $c(r,V)$ is the expected loss brought by sensitive data leakage,
defined as $c(r,V)=\alpha(1-r)+\beta V$,  which has the similar definition with $\tilde{c}(r,V)$ and $\alpha, \beta\geq 0$ are coefficients. 
With a higher value of the sensitive data,
the markup price  proposed by the data provider, namely $D_s$, should also be higher. So we define $V=\eta D_s$, with $\eta >0$ being the coefficient.

\begin{figure}[ht]
\centerline{
\includegraphics[width=6cm, height=2cm]{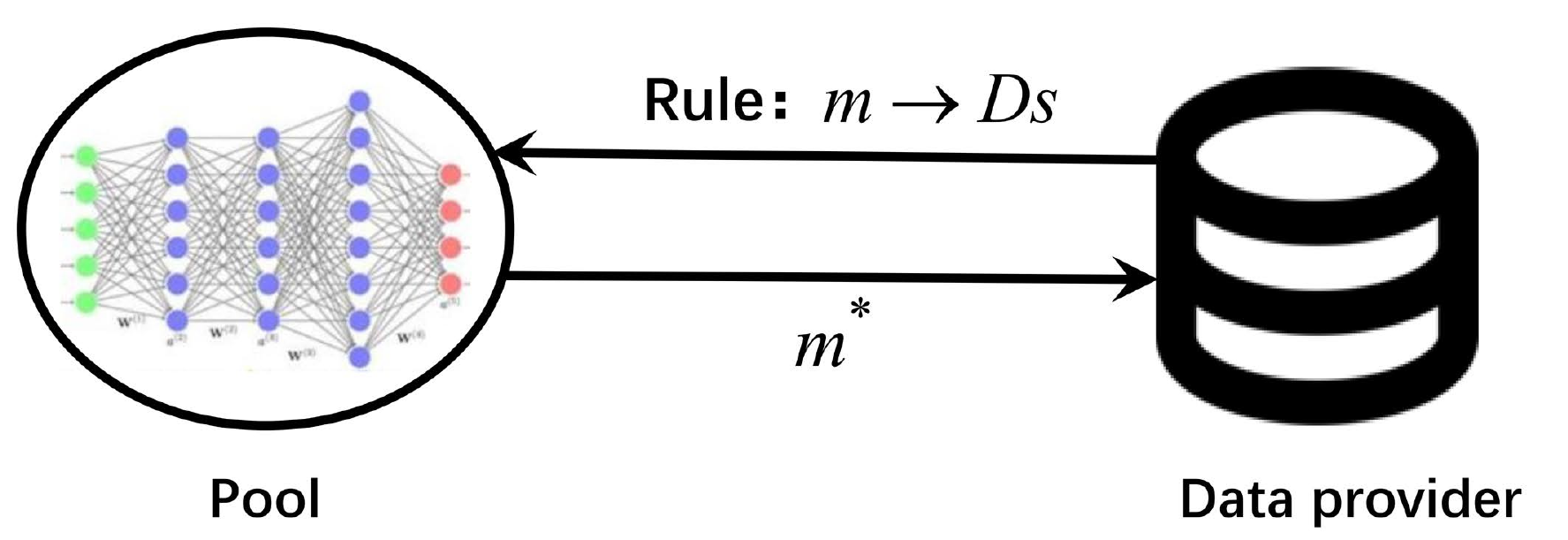}}
\caption{Framework of reputation-based data trading mechanism.}
\label{MD}
\end{figure}

In the reverse game-based data trading mechanism, the strategies of the pool and the data provider are respectively $m$ and $D_s(m)$. In other words, the  strategy of the data provider is not a value but a rule, i.e., a function, which empowers the data provider to force the pool to make the optimal bid based on his real private information, such as
the legally expected net profit ($Q$) of a pool from this data trading and his expected profit ($\tilde{c}(r,V)$) from  leaking sensitive data.
Our proposed mechanism is illustrated in Fig. \ref{MD}, which includes three phases:

\begin{itemize}
\item Phase 1: The data provider first designs an optimal game rule $D_s^*(m)$ which can maximize her utility.
\item Phase 2: Once receiving $D_s^*(m)$, the pool needs to decide whether or not to accept the game rule. If yes, he will calculate the optimal bid $m^*$ according to $D_s^*(m)$ for maximizing his utility. Otherwise, he just ignores the received message.
\item Phase 3: If the data provider does not receive $m^*$ before a given deadline, the trading negotiation is terminated. Otherwise,  she will calculate $D_s^*$ and thus the probability of data trading in this round, i.e., $p$, can also be derived in light of \eqref{p}. Once the data provider and the pool reaches an agreement on the data trading, the final price of the training data is $m^*+D_s^*$.
\end{itemize}

In the following, we will introduce how to obtain the optimal strategies of both parties, i.e., $D_s^*(m)$  and $m^*$. Let $F_{d}=p\cdot(m+D_s-c(r,V))=[\varepsilon_1 r+\varepsilon_2 \frac{D_s}{\bar{D_s}}+(1-\varepsilon_1-\varepsilon_2)\frac{m}{\bar{m}}]\cdot[m+D_s-\alpha(1-r)-\beta \eta D_s]$, the integrand of \eqref{provider_benefit}. To maximize $U_{provider}$, we adopt the  variational method. In detail,  through solving the Euler-Lagrange equation $\frac{\partial F_{d}}{\partial D_s}-\frac{\mathrm{d}}{\mathrm{d}r}\frac{\partial F_{d}}{\partial D_s^{'}}=0$  under $\frac{\partial^2F_{d}}{\partial D_s^2}=\frac{2 \varepsilon_2 (1-\beta \eta)}{\bar{D_s}}<0$, we have
\begin{small}
\begin{equation}
\label{eq:Ds}
D_s^*(m)=\frac{\varepsilon_2[m-\alpha(1-r)]+(1-\beta \eta)\bar{D_s}[\varepsilon_1 r+(1-\varepsilon_1-\varepsilon_2)\frac{m}{\bar{m}}]}{2\varepsilon_2(\beta \eta-1)}.
\end{equation}
\end{small}

Similarly,
let the integrand of \eqref{up} be $F_{p}=p\cdot(Q+\tilde{c}(r,V)-m-D_s)=[\varepsilon_1 r+\varepsilon_2 \frac{D_s^{*}}{\bar{D_s}}+(1-\varepsilon_1-\varepsilon_2)\frac{m}{\bar{m}}] \cdot[Q+\tilde{\alpha}(1-r)+\tilde{\beta} \eta D_s^{*}-m-D_s^{*}]$. The optimal strategy of  the pool can also be calculated through the  variational method. That is,
\begin{small}
\begin{equation}
\label{eq:m}
\begin{split}
&m^*=\\
&\frac{A_0A_2A_3\varepsilon_1 r+(A_0A_2\varepsilon_2+\bar{D_s}A_0^{2}A_1)[Q+\tilde{\alpha}(1-r)]-\varepsilon_1 r\bar{D_s}A_0^{2}}{2\varepsilon_2 A_0A_2+2A_1A_0^{2}\bar{D_s}-2\varepsilon_2(\tilde{\beta} \eta-1)A_2^{2}-2A_0A_1A_2A_3}  \\
&+\frac{[2\varepsilon_2 A_2A_3+A_0A_1A_3-A_0\varepsilon_2][(1-\beta \eta)\bar{D_s}\varepsilon_1 r-\varepsilon_2 \alpha(1-r)]}{2\varepsilon_2 A_0A_2+2A_1A_0^{2}\bar{D_s}-2\varepsilon_2(\tilde{\beta} \eta-1)A_2^{2}-2A_0A_1A_2A_3},
\end{split}
\end{equation}
\end{small}
in which
\begin{align}
&A_0=2\varepsilon_2(\beta \eta-1),
A_1=(1-\varepsilon_1-\varepsilon_2)\frac{1}{\bar{m}},     \notag \\
&A_2=\varepsilon_2+(1-\beta \eta)(1-\varepsilon_1-\varepsilon_2)\frac{\bar{D_s}}{\bar{m}}, A_3=\bar{D_s}(\tilde{\beta} \eta-1). \notag
\end{align}

\begin{theorem}[]
\label{thm:Ds}
When $1-\beta \eta<0$, $D_s^*(m)$ in (\ref{eq:Ds}) is the equilibrium strategy of the data provider.
\end{theorem}
\begin{proof}
According to the variational method, when $\frac{\partial^2F_d}{\partial D_s^2}-\frac{\mathrm{d}}{\mathrm{d}t}\frac{\partial^2F_d}{\partial D_s \partial D_s^{'}} \leq 0$ and $\frac{\partial^2F_d}{\partial (D_s^{'})^2}\leq 0$, $F_d$ can be maximized. Because $F_d$ is not related to $D_s^{'}$, $\frac{\mathrm{d}}{\mathrm{d}t}\frac{\partial^2F_d}{\partial D_s \partial D_s^{'}}=0$ and $\frac{\partial^2F_d}{\partial (D_s^{'})^2}=0$.  Thus, the condition of maximizing $F_d$ is simplified to $\frac{\partial^2F_d}{\partial D_s^2}\leq 0$, implying   $\frac{2 \varepsilon_2 (1-\beta \eta)}{\bar{D_s}}\leq 0$ should be met. However,  if  $\frac{2 \varepsilon_2 (1-\beta \eta)}{\bar{D_s}}=0$, $D_s^*(m)$ in (\ref{eq:Ds}) will be meaningless, so $\frac{\partial^2F_d}{\partial D_s^2}=\frac{2 \varepsilon_2 (1-\beta \eta)}{\bar{D_s}}< 0$ should be satisfied. Due to  $\varepsilon_2>0$ and $\bar{D_s}>0$, when $1-\beta \eta<0$ is satisfied, $F_d$ can be maximized. Thus, the theorem is proved.
\end{proof}

By the similar way, we can obtain the following theorem:
\begin{theorem}[]
\label{thm:m}
When $1-\tilde{\beta} \eta<0$, $m^*$ in  (\ref{eq:m}) is the equilibrium strategy of the pool.
\end{theorem}

\begin{theorem}[]
\label{thm:IC}
The game rule $D_s^*(m)$ designed by the data provider  is incentive-compatible.
\end{theorem}
\begin{proof}
 We assume that $\hat{Q}$ and $\widehat{\tilde{c}(r,V)}$ are the fake private information, based on which a dishonest  strategy $\hat{m}$ is reported to the data provider. Due to $\hat{Q}\neq Q$ and $\widehat{\tilde{c}(r,V)}\neq \tilde{c}(r,V)$, $\hat{m}\neq m^*(Q,\tilde{c}(r,V))$. Because  $U_{pool}$ is maximized  only when his strategy is $m^*(Q,\tilde{c}(r,V))$  according to (\ref{eq:m}), $U_{pool}(\hat{m}(\hat{Q},\widehat{\tilde{c}(r,V)})) \leq U_{pool}(m^*(Q,\tilde{c}(r,V))$. Thus, the pool can maximize his utility only if he reports the strategy based on the true privacy information. In other words, the game rule is incentive compatible.
\end{proof}

The game rule satisfying the incentive-compatibility principle drives the pool to calculate $m^*$ based on his real private information. As the strategy of the data provider is a function of $m^*$, i.e., $D_s(m^*)$, it will also be derived based on real private information. This is equivalent to the situation where both the pool and the data provider make the optimal strategies for maximizing their  utilities based on the global information known by two sides. In addition, the incentive-compatibility of the game rule enforces both $m^*$ and $D_s(m^*)$ to reveal the risk of the data privacy leakage in this round of data trading. This is because they are calculated based on the real private information of the pool, i.e., his legally expected net profit ($Q$) from this data trading and expected extra profit ($\tilde{c}(r,V)$) from leaking sensitive data, while both rational players in the reverse game are  utility-driven, making $m^*$ and $D_s(m^*)$ closely related to whether the data traded in this round will be leaked or not. Therefore, the calculation of successful trading probability in \eqref{p} can reduce or even prevent privacy leakage behavior since it depends on not only the historical data privacy leakage behavior of the pool  but also  the privacy leakage risk in this round.

\section{Privacy-preserving Model Verification Mechanism}
\label{sec:2PC}

As we mentioned above, after each epoch of training, the model accuracy  should be calculated based on the test data of the requester. However, on one hand, the test data may be sensitive so that the requester is not willing to share them; on the other hand,  the trained model should not be published before the end of consensus competition to avoid being  plagiarized by other competitors. To address this challenge, we design a privacy-preserving model verification mechanism to verify the accuracy of the trained model without information disclosure from either the requester or the pool. This mechanism can also be used by full nodes to verify the accuracy of models when receiving blocks from pools.

The accuracy of the trained model is evaluated by  the number of the same predicted labels as the actual ones. Hence, the model verification in our mechanism is divided into two parts, i.e., label prediction and label comparison.
To illustrate how our mechanism accomplishes label prediction and comparison, we take the deep feedforward network model\footnote{Other models can be processed in the same way.} as an example, which is a typical type of deep learning network structure.

 \subsection{Homomorphic encryption-based label prediction}
    According to the design of the deep feedforward network, 
the value of a node $i$ in each layer is calculated based on the inputs $\mathbf{x}=\{x_j\}$ from the last layer, which is $f(\mathbf{x})=\sigma(z_i)$. Here, $\sigma(\cdot)$ represents the sigmoid function and $z_i=\mathbf{x}^T\mathbf{w}_i+b_i$, where $\mathbf{w}_i=\{w_{ij}\}$ is the weight vector and $b_i$  is the bias of node $i$, with $j$ denoting all nodes in the last layer.
Particularly, the input of the first layer is the test data of the requester $A=[a_{ij}]_{I\times M}$, where $a_{ij}$ is the $j$-th attribute in the $i$-th row of data with $1\leq i\leq I$ and $1 \leq j \leq M-1$ with
$I$  indicating how many pieces\footnote{For instance, one piece of data in a medical dataset denotes the data of one patient.}
of test data that the requester owns and $M-1$ denoting the number of attributes in each piece of data; while the last item of each piece of test data, i.e., $a_{iM}$  ($1\leq i\leq I$), is the label.

Since $A=[a_{ij}]_{I\times M}$ is sensitive for the requester, to avoid leaking $A$ in the process of calculating the accuracy of models, an effective method is secure two-party computation (2PC). For example, Rouhani et al. \cite{Rouhani18} realized privacy-preserving label prediction utilizing 2PC. However, the communication and computation complexity of directly using 2PC is too high, making the efficiency of label prediction relatively low. To address this problem, we propose the homomorphic encryption (HE)-based label prediction. In detail, the requester only needs to send the encrypted $A$, denoted as $[\left \langle a_{ij}\right \rangle]_{I\times M}$, as well as the corresponding public key to the pool, who can use HE to calculate the outputs of the second layer, i.e., $f(\mathbf{x})$, based on the first (input) layer of encrypted data attributes. This is because the HE technology can output the desired calculation results in plaintext with operations on the ciphertext.
Assuming that 
the number of nodes in the first layer is the same as the number of attributes, i.e., $M-1$, and the number of nodes in the second layer is represented by $K$. 
Take the first piece of data as an example, we can calculate the value of node $k$ in the second layer as:
  \begin{equation}
\left \langle z_{k}\right \rangle=\sum_{j=1}^{M-1}({\left \langle a_{1j}\right \rangle} \otimes {\left \langle w_{kj}\right \rangle}\oplus \left \langle b_k \right \rangle), \  1\leq k \leq K,
\label{e1}
\end{equation}
in which $\oplus$ and $\otimes$ are the corresponding computation of $+$ and $\times$ under HE, and $\sum$ is also operated with $\oplus$ calculation.

In order to prevent disclosing the real value of $z_{k}$ ($1\leq k \leq K$) from the requester, the pool masks them with a random vector $[h_{k}]_K$, which can be encrypted to $[\left \langle h_{k} \right \rangle]_K$ with the requester's public key. So $\left \langle z_{ik}+h_k \right \rangle$ is calculated as follows:
\begin{equation}
\left \langle z_{k}+h_k \right \rangle= \left \langle z_{k} \right \rangle \oplus  \left \langle h_k \right \rangle, \ \   1\leq k \leq K.
\end{equation}
Once  the requester receives  $\left \langle z_{k}+h_k \right \rangle$, she decrypts them with the private key and sends back $z_{k}+h_k$ $(1\leq k \leq K)$ in plaintext to the pool, which helps to calculates $\sigma(z_{k})$ serving as the input of the next layer in  the deep feedforward network. After that, the rest layers in the deep feedforward network model can be calculated locally by the pool with no need to interact with the requester until deriving the predicted labels.

 \subsection{2PC-based label comparison}
After label prediction, label comparison starts, where the accuracy of the model can be derived through counting the number of the same predicted label $a'_{iM}$ as the actual one $a_{iM}$ ($1\leq i \leq I$). However, there still exists the privacy protection issue in this process. On one hand, the pool is not willing to disclose the predicted label $a'_{iM}$ since the trained model might be inferred from this side information. On the other hand, the actual label $a_{iM}$ is sensitive data for the requester as we mentioned above. To overcome this challenge, we propose a 2PC-based label comparison. The main reason of using 2PC lies in that  it is easy to employ a garbled circuit (GC)  to describe the label comparison process so that this process can be completed efficiently, since the efficiency of 2PC is closely related to the GC construction in both the communication and computation complexity \cite{Mihir13} \cite{Kolesnikov08} \cite{Pinkas09}.

 To realize the proposed method, we firstly design  the boolean circuit for comparison as shown in Fig. \ref{label_compare1}. The inputs of the circuit are $a'_{iM}$ and $a_{iM}, ~1 \leq i \leq I$, both of which are assumed to be $l$-bit,  and the output is the total number $N$ of correct predictions. Based on our circuit design, if $a_{iM}=a'_{iM}$, the statistical total number $N$ of correct prediction adds 1. Otherwise, this is an erroneous one and $N$ adds 0. Even though the accuracy should be calculated by $N/I$, we use the total number of correct predictions $N$ to denote it for simplicity since $I$ is stable in one round of calculation, verification and comparison.

\begin{figure}[h]
\centerline{
\includegraphics[width=9cm, height=6.5cm]{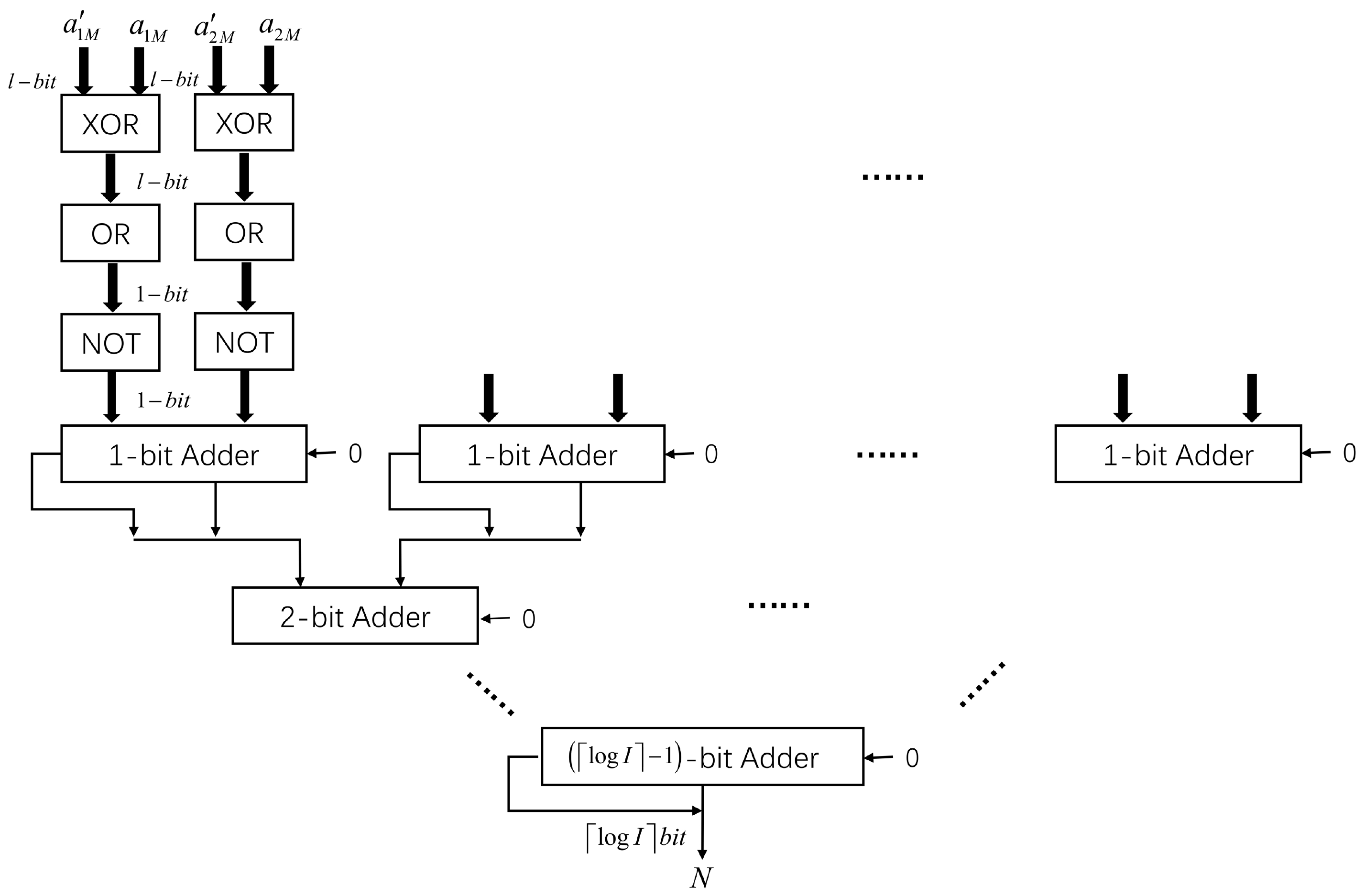}}
\caption{Circuit for label comparison.}
\label{label_compare1}
\end{figure}

The cost of GC is linearly correlated with the number of garbled gates \cite{Huang11}. Taking advantage of  the {\it free-XOR} technology \cite{Kolesnikov08}, XOR can be {\it free}, implying that XOR does not need associated garbled tables
and the corresponding hashing or symmetric key operations.  Therefore, a direct way to improve the performance is reducing the number of costly garbled gates,  namely those non-XOR gates. In detail,  to compare two values, a SUB gate is supposed to be exploited, which is a non-XOR gate.  In our scheme, we replace the SUB gate with a combination of an XOR gate and a NOT gate shown in Fig. \ref{label_compare2}. Note that the NOT gate is also free since it can be implemented using an XOR gate with one of the inputs as constant 1.  We summarize the number of non-free binary gates of our circuit in Table \ref{tab3}, where the $n$-bit Adder denotes all the adders from 1 bit to $\left \lceil log I \right \rceil-1$ bit, namely $n \in \{1, \cdots, \left  \lceil log I \right \rceil-1\}$.
\begin{figure}[h]
\centerline{
\includegraphics[width=4cm, height=3cm]{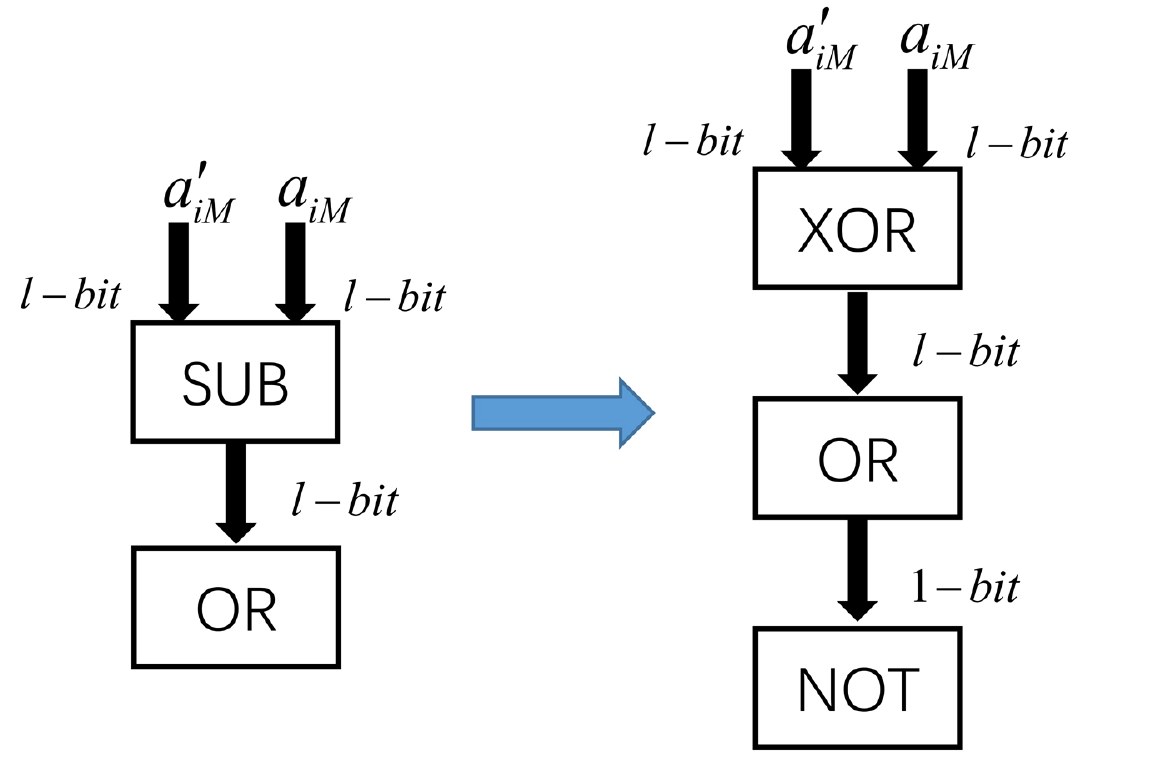}}
\caption{Optimization of the non-free SUB gate.}
\label{label_compare2}
\end{figure}
\begin{table}[ht]
\begin{center}
\caption{Number of Non-Free Gates}
\begin{spacing}{1}
\begin{tabular}{|c|c|c|c|}
\hline
Gate & OR & n-bit Adder & All Gates                                         \\ \hline
Number & $I$  & $\frac{I \left \lceil log I \right \rceil}{2}$ & $I+\frac{I \left \lceil log I \right \rceil}{2}$ \\ \hline
\end{tabular}
\end{spacing}
\label{tab3}
\end{center}
\end{table}

After the fundamental construction of the boolean circuit, we employ JustGarble \cite{Mihir13} for garbling, due to its optimization in high efficiency, proven security and garbled row reduction. In short, to garble a circuit is to encrypt the inputs and outputs of the gates in the circuit and disorder the permutation of them. We omit the description of detailed operation for garbling because they are simple variable operations like a couple of shifts. It is worth mentioning that the time cost of the above garbling way is in the level of nanoseconds per gate, and the size of garbled tables is $10^1$ order of magnitude, which is efficient enough to meet the requirement for model verification.

Note that GC is finally composed of garbled tables, which gives the keys used to encrypt the inputs and output of each gate. After  constructing  GC, the pool sends it to the requester along with the corresponding keys of predicted labels, denoted as $w(a'_{iM})$ ($1\leq i \leq I$). In order to verify the model accuracy without leaking the privacy of each other, the oblivious transfer  (OT) \cite{Harnik08}\cite{Stanislaw09} is adopted, by which
the requester can only find out the corresponding keys of her actual labels $w(a_{iM})$ ($1\leq i \leq I$) without any knowledge on the predicted labels, so neither the pool can know anything about the actual labels. 
After receiving the keys of inputs $a'_{iM}$ and $a_{iM}$, the requester can calculate the corresponding garbled encrypted-output and evaluate GC to find out the actual output $\left \langle N\right \rangle$, which is the encrypted value of $N$. Once receiving $\left \langle N\right \rangle$, the pool decrypts it and  packages it in the block as accuracy.

It is worth noting that our model verification mechanism can also be utilized by full nodes to verify a model's accuracy. Therefore, the encrypted weights and biases need to  be packaged in the block with  $GC$. All of these encrypted data are denoted by $V_m$, which are stored in the block header as we mentioned in Section \ref{sec:flow}.

\section{Experimental Evaluation}
\label{sec:experiment}
In this section, we evaluate our proposed PoFL through  simulations based
on synthetic and real-world data as follows.
\subsection{Data trading}
In this subsection, we study the impact of some key parameters in our proposed data trading mechanism. Firstly, we study the impact of the pool's reputation $r$. Here we set $\varepsilon_1=\varepsilon_2=0.4$, $\eta=1.8$, $\alpha=\tilde{\alpha}=1.5$, $\beta=\tilde{\beta}=1$, and $Q=8$ to satisfy Theorems \ref{thm:Ds} and \ref{thm:m}\footnote{Other  parameters satisfying requirements are also evaluated, which presents similar performance trend. So we omit these results to avoid redundancy.}. As shown in Fig. \ref{IC_exp_1}(a), when $r$, the maximum utilities of both the pool and the data provider quadratically increase. The similar trend happens on   the probability that the pool can purchase training data as illustrated in Fig. \ref{IC_exp_1}(b). The reason behind these facts is that the higher $r$ will bring a higher data trading probability $p$ in light of \eqref{p} and a lower expected loss  brought by sensitive data leakage $c(r,V)$ according to \eqref{provider_benefit}, which together contribute to the increase of utilities for both sides.

Then we examine the impact of the pool's legally expected net profit $Q$, which is reported in Fig. \ref{IC_exp_1}(c)(d). Note that the parameter setting here is the same as the above ones except the varying $Q$ and $r=0.5$.
It can be observed from figures that both maximum utilities and the data trading probability  increase with $Q$. This is because the larger $Q$, the higher the data price that the pool is willing to provide, which not only increases  $p$ as defined in \eqref{p} but also improves the utilities for both sides due to \eqref{up} and \eqref{provider_benefit}.

\begin{figure}[h]
\subfigure[Utility vs. $r$]{
\centering
\includegraphics[width=0.22\textwidth]{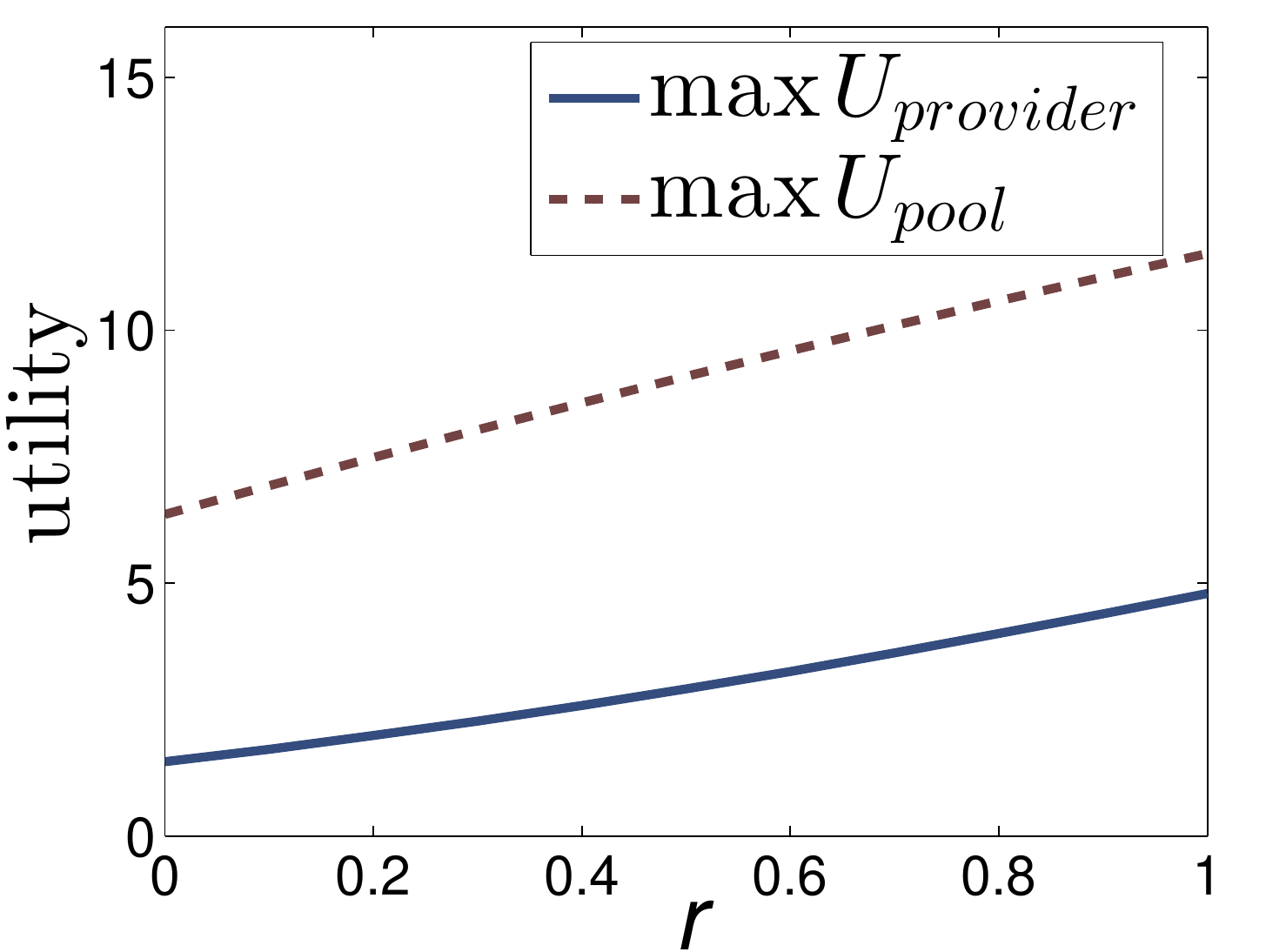}
}
\subfigure[Probability vs. $r$]{
\centering
\includegraphics[width=0.22\textwidth]{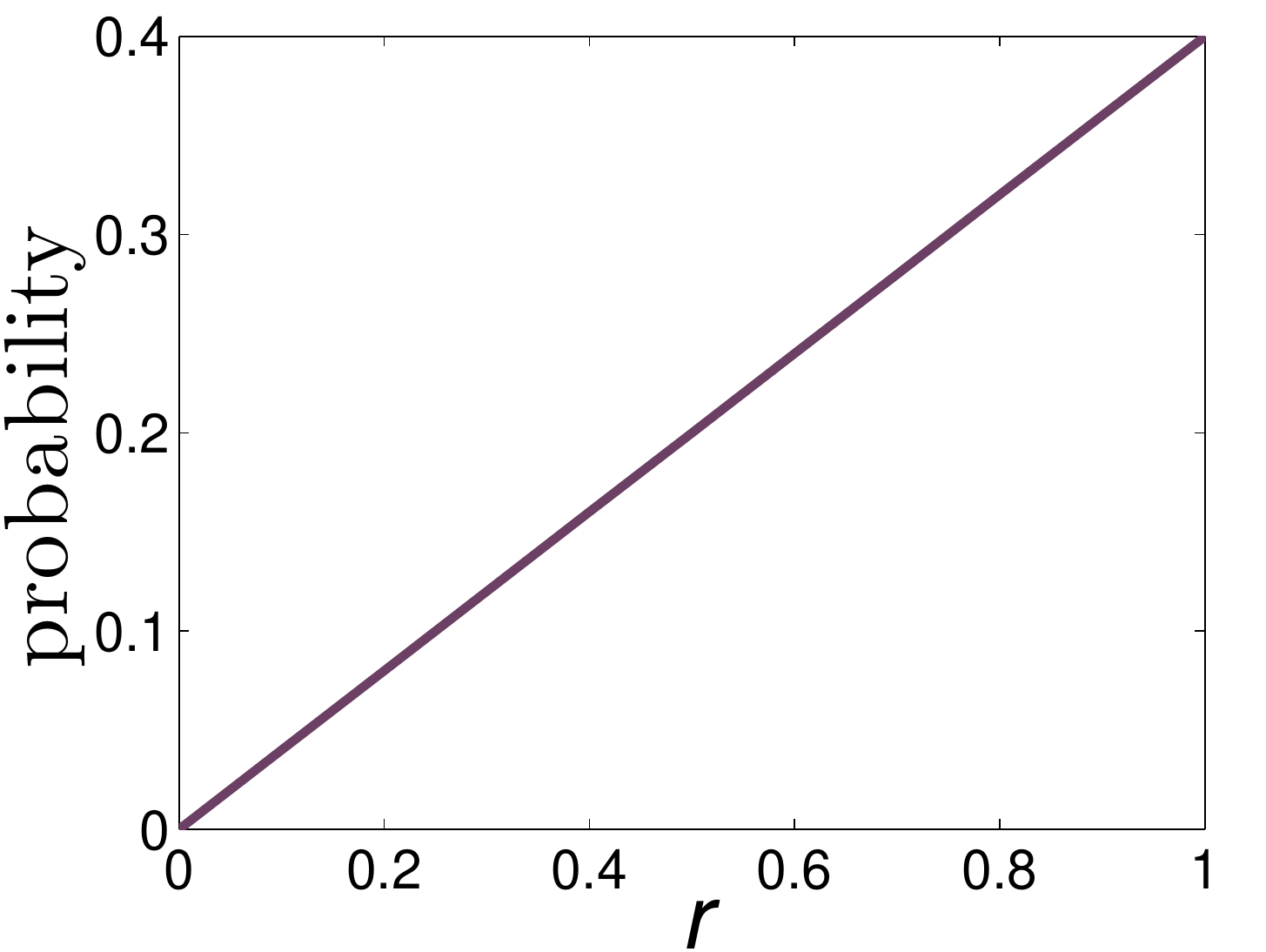}
}
\subfigure[Utility vs. $Q$]{
\centering
\includegraphics[width=0.22\textwidth]{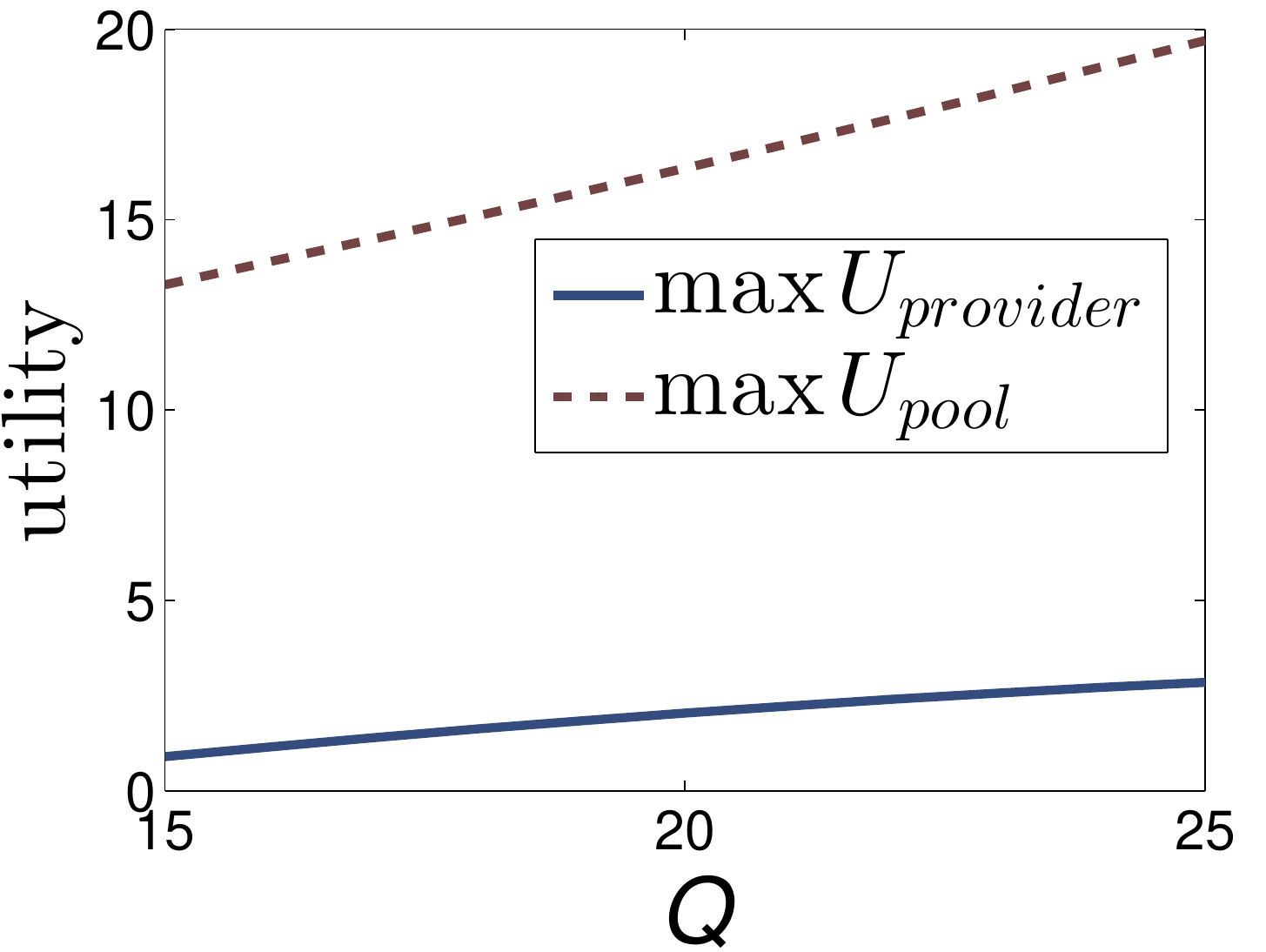}
}
\subfigure[Probability vs. $Q$]{
\centering
\includegraphics[width=0.22\textwidth]{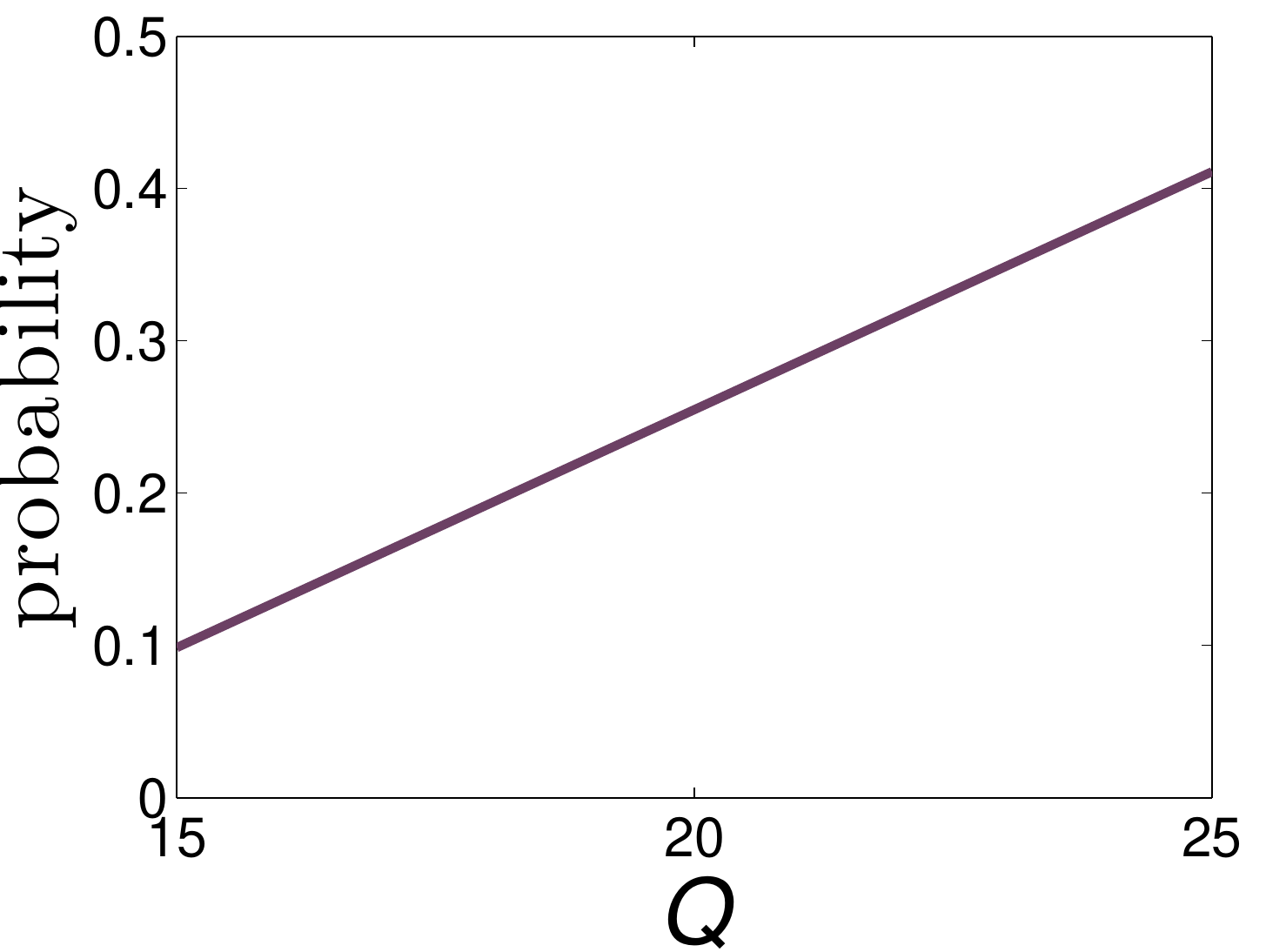}
}
\caption{Impacts of $r$ and $Q$ on data trading.}
\label{IC_exp_1}
\end{figure}

Next, we evaluate the impacts of $\tilde{\alpha}$ and $\tilde{\beta}$ on data trading, which directly reflect the expected profit of leaking private data, i.e., $\tilde{c}(r,V)$.
With the same $\varepsilon_1, ~\varepsilon_2$ mentioned above, we use $\eta=1.8,~\beta=\tilde{\beta}=1,~Q=20$ in Fig. \ref{IC_exp_2}(a)(b) and    $\eta=0.5,~\alpha=\tilde{\alpha}=5,~Q=20$ in Fig. \ref{IC_exp_2}(c)(d).
As shown in Fig. \ref{IC_exp_2}(a)(c), with the increase of $\tilde{\alpha}$ and $\tilde{\beta}$, which is equivalent to the increase of $\tilde{c}(r,V)$, the maximum utility of the pool increases due to his selling of sensitive data, while the utility of the data provider decreases because of her increasing loss brought by data leakage.
According to Fig. \ref{IC_exp_2}(b)(d), one can tell that as the extra profit that the pool can gain from leaking data increases, the probability he can obtain the data is decreasing significantly, which can act as a reverse driving force for the legal behavior of the pool.
\begin{figure}[h]
\subfigure[Utility vs. $\tilde{\alpha}$]{
\centering
\includegraphics[width=0.22\textwidth]{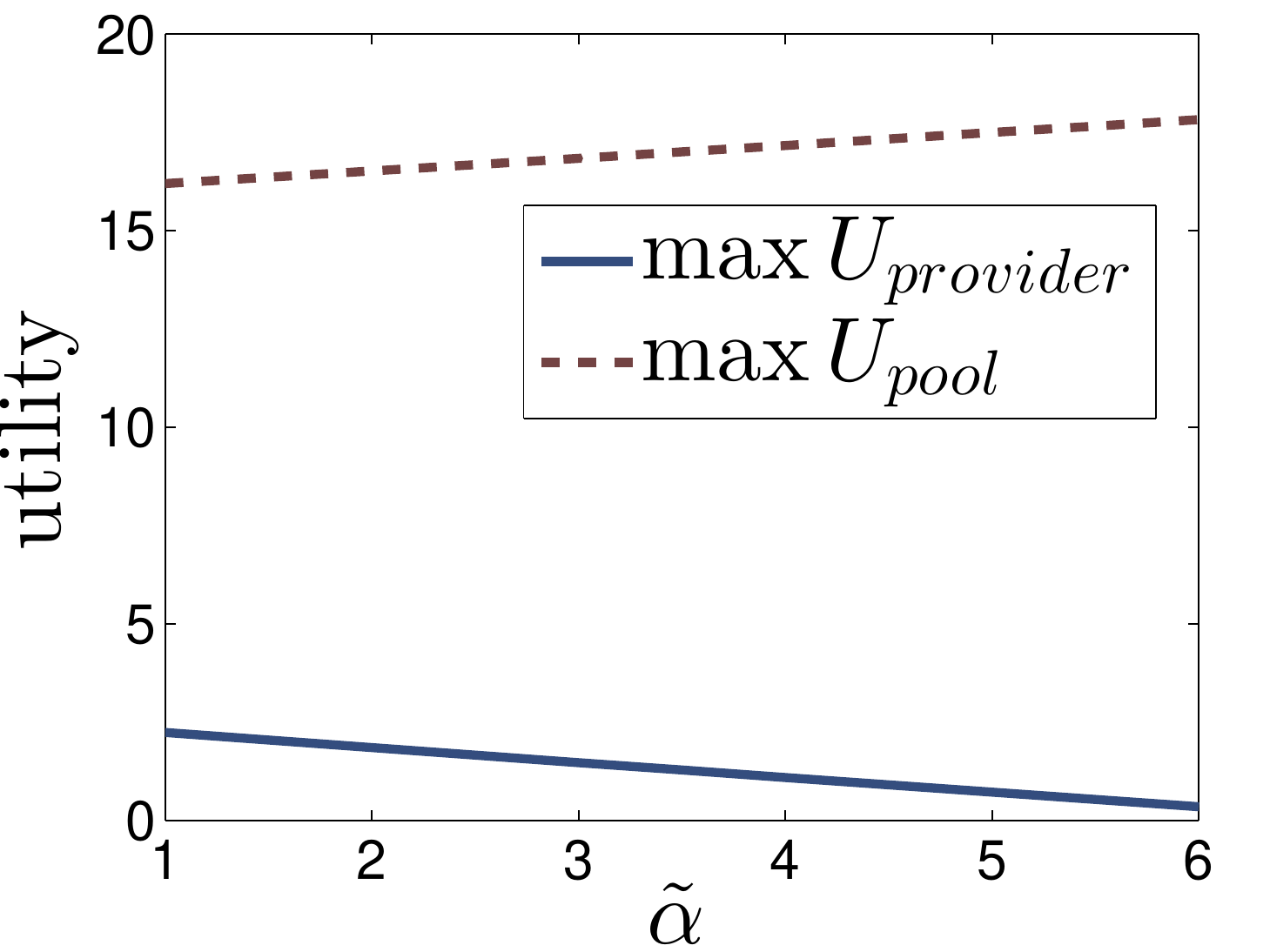}
}
\subfigure[Probability vs. $\tilde{\alpha}$]{
\centering
\includegraphics[width=0.22\textwidth]{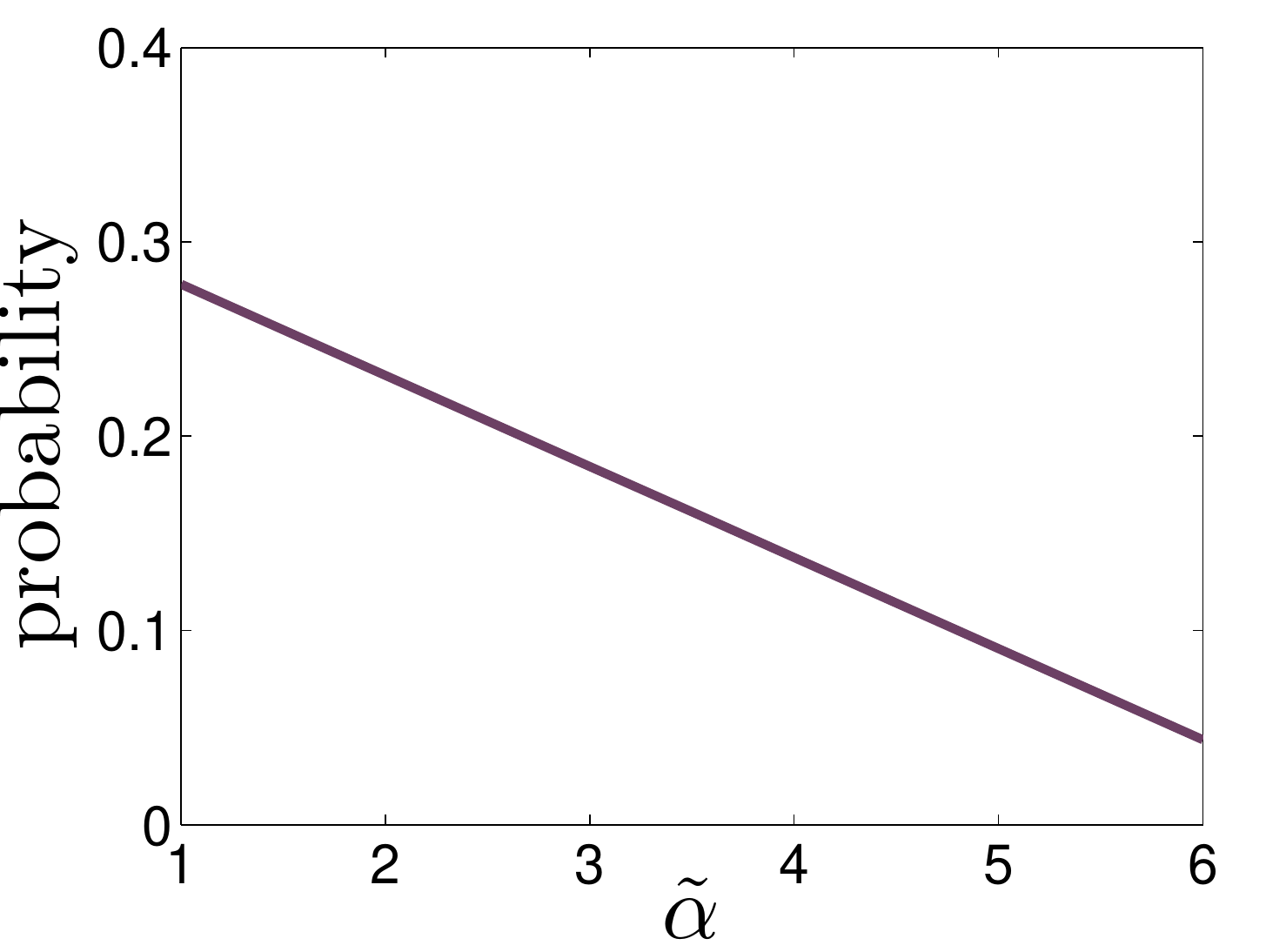}
}
\subfigure[Utility vs. $\tilde{\beta}$]{
\centering
\includegraphics[width=0.22\textwidth]{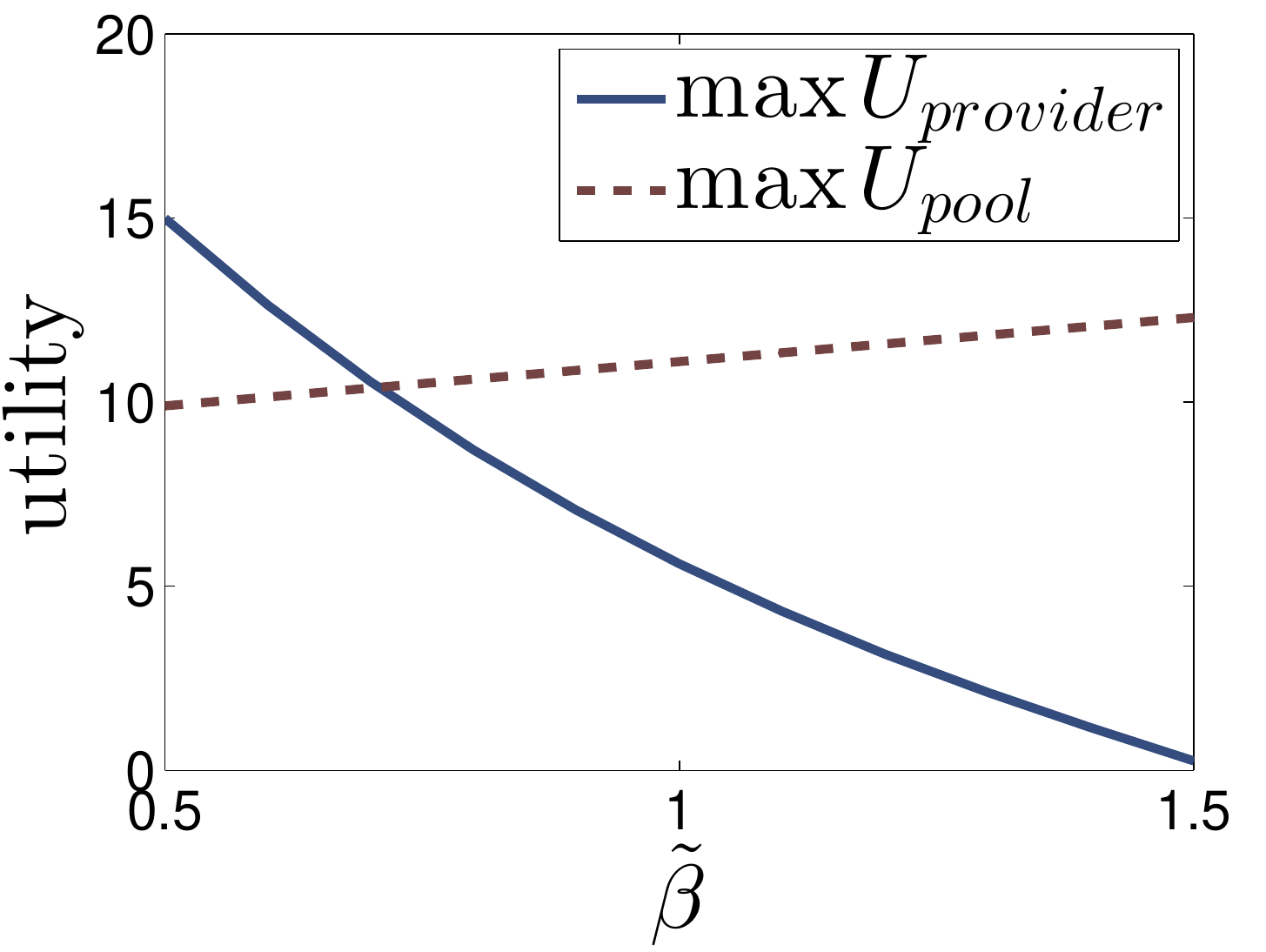}
}
\subfigure[Probability vs. $\tilde{\beta}$]{
\centering
\includegraphics[width=0.22\textwidth]{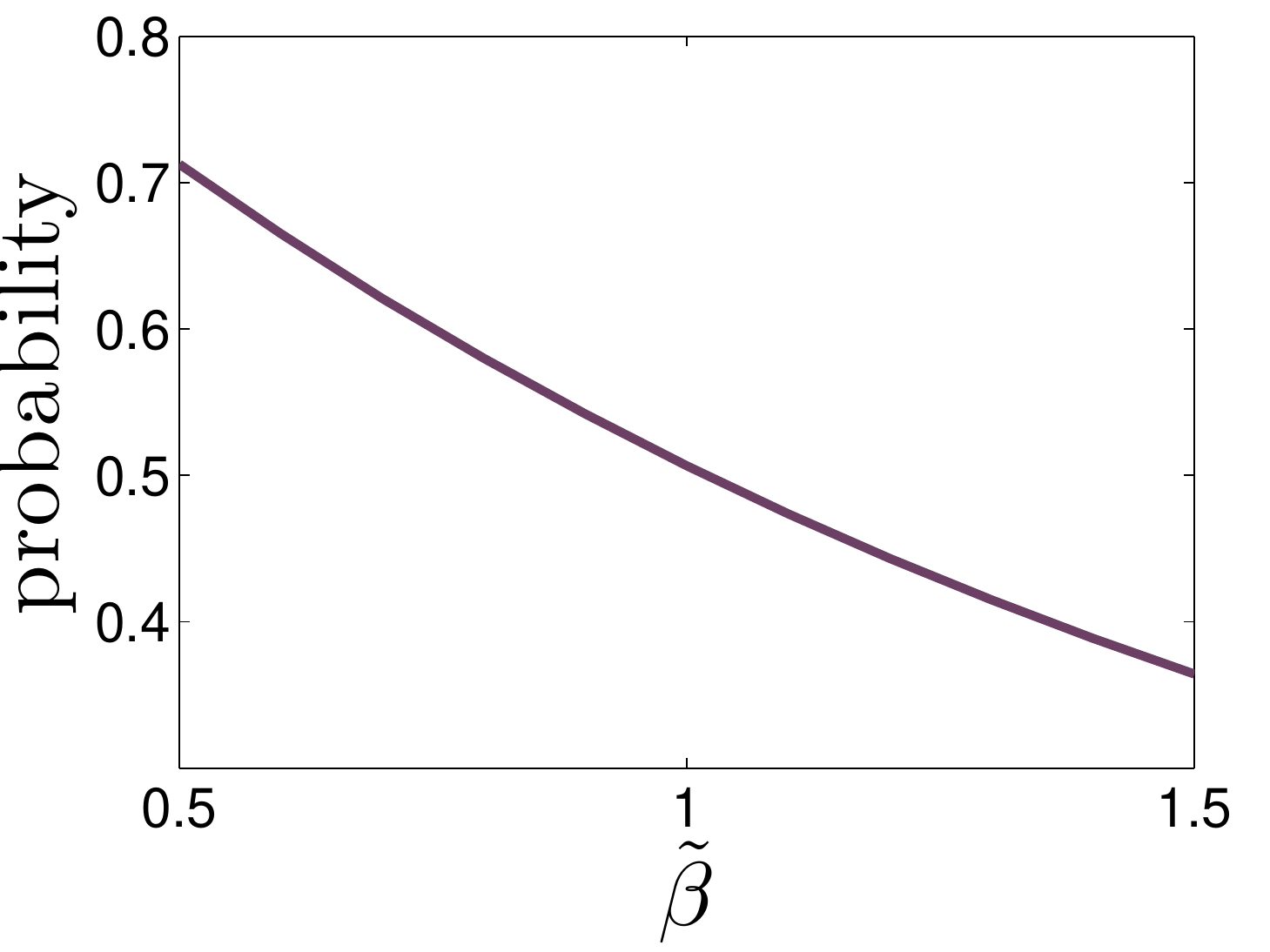}
}
\caption{Impacts of $\tilde{\alpha}$ and $\tilde{\beta}$ on data trading.}
\label{IC_exp_2}
\end{figure}

\subsection{Federated mining}
We simulate the federated mining process based on the CIFAR-10 dataset \cite{Krizhevsky09} with 50,000 training samples which is composed of 10 classes of $32 \times 32$ images with three RGB channels.
Here we set that the initial model selected by the pool manager is AlexNet \cite{Krizhevsky12}. Even though there are other models with higher accuracy like more than 96.53\% \cite{Benjamin14}, AlexNet can function well as an example in our proposed mechanism. In fact, how to select the initial model is out of the scope of our work.

To be specific, we take the linear regression training as an example\footnote{Other training functions can be implemented in a similar way.}. We use Batch Gradient Descent (BGD) to optimize the parameters of the model. Assuming that the total number of miners contributing to federated mining in the pool is $\mathcal{K}$ and the learning rate is $\zeta$.
Fig. \ref{lab1}(a) illustrates, when the learning rate $\zeta$ changes, the convergence speed to achieve certain accuracy is different but it is not the case where the higher the learning rate the better. 
This is because $\zeta$ indicates the length of step in the direction of gradient. If $\zeta$ is too small, the training process can be time-consuming; while if it is too large, it may cause excessive loss as the gradient obtained from the training data is an approximate value of the real one. 

Assuming that the whole dataset is evenly distributed among miners in the pool and each miner has the same proportion $S$ of the dataset, which is equivalent to $\frac{1}{\mathcal{K}}$. It turns out that the larger the portion of data each miner owns, the faster to reach the highest accuracy, shown in Fig. \ref{lab1}(b). This is because with more data, the calculated gradient will be closer to the real value for the fastest gradient descent.

\begin{figure}[ht]
\centering
\subfigure[Learning Rate]{
\begin{minipage}[t]{0.45\linewidth}
\centering
\includegraphics[width=1.0\textwidth]{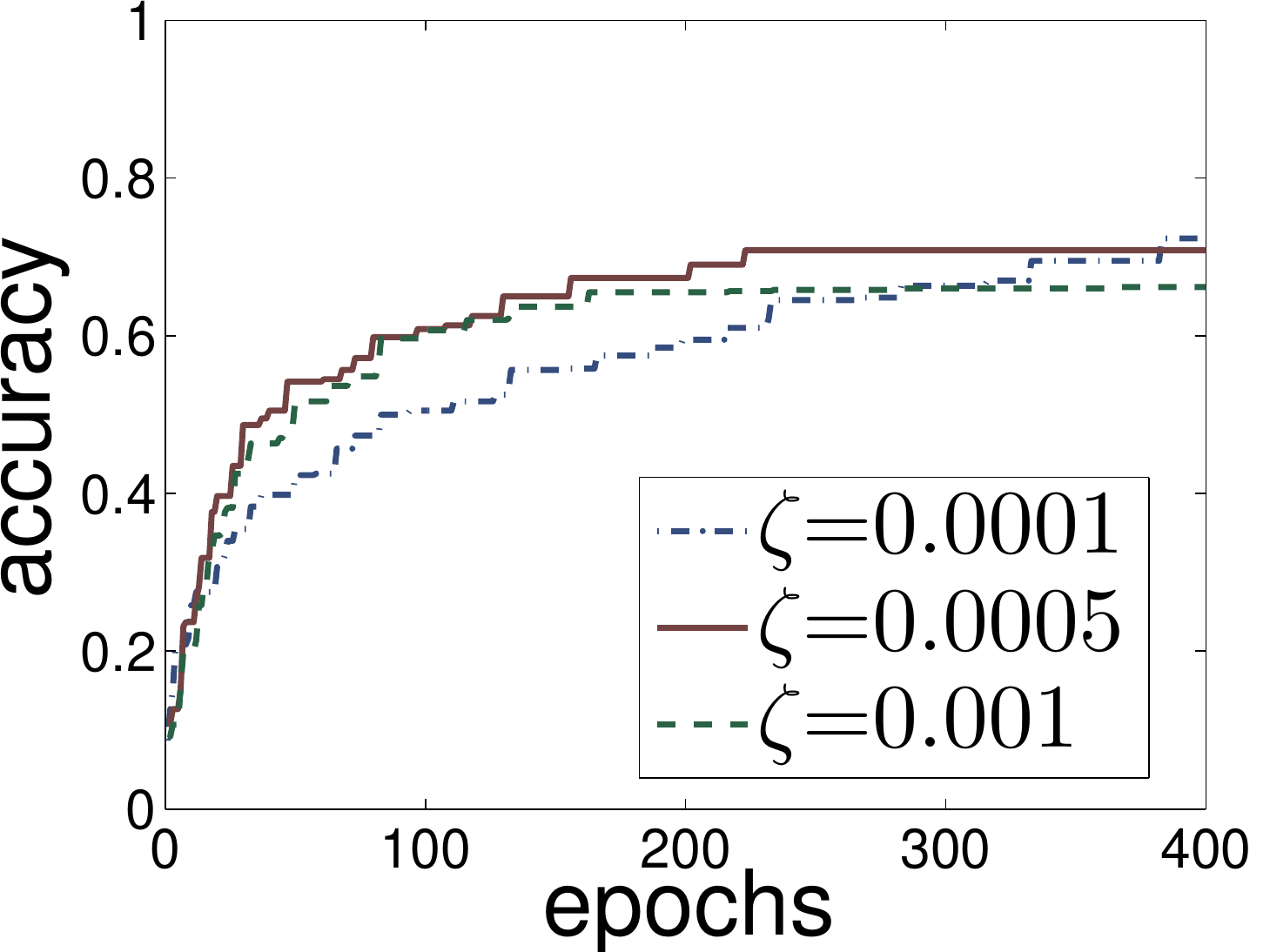}
\end{minipage}
}
\subfigure[Data Quantity]{
\begin{minipage}[t]{0.45\linewidth}
\centering
\includegraphics[width=1.0\textwidth]{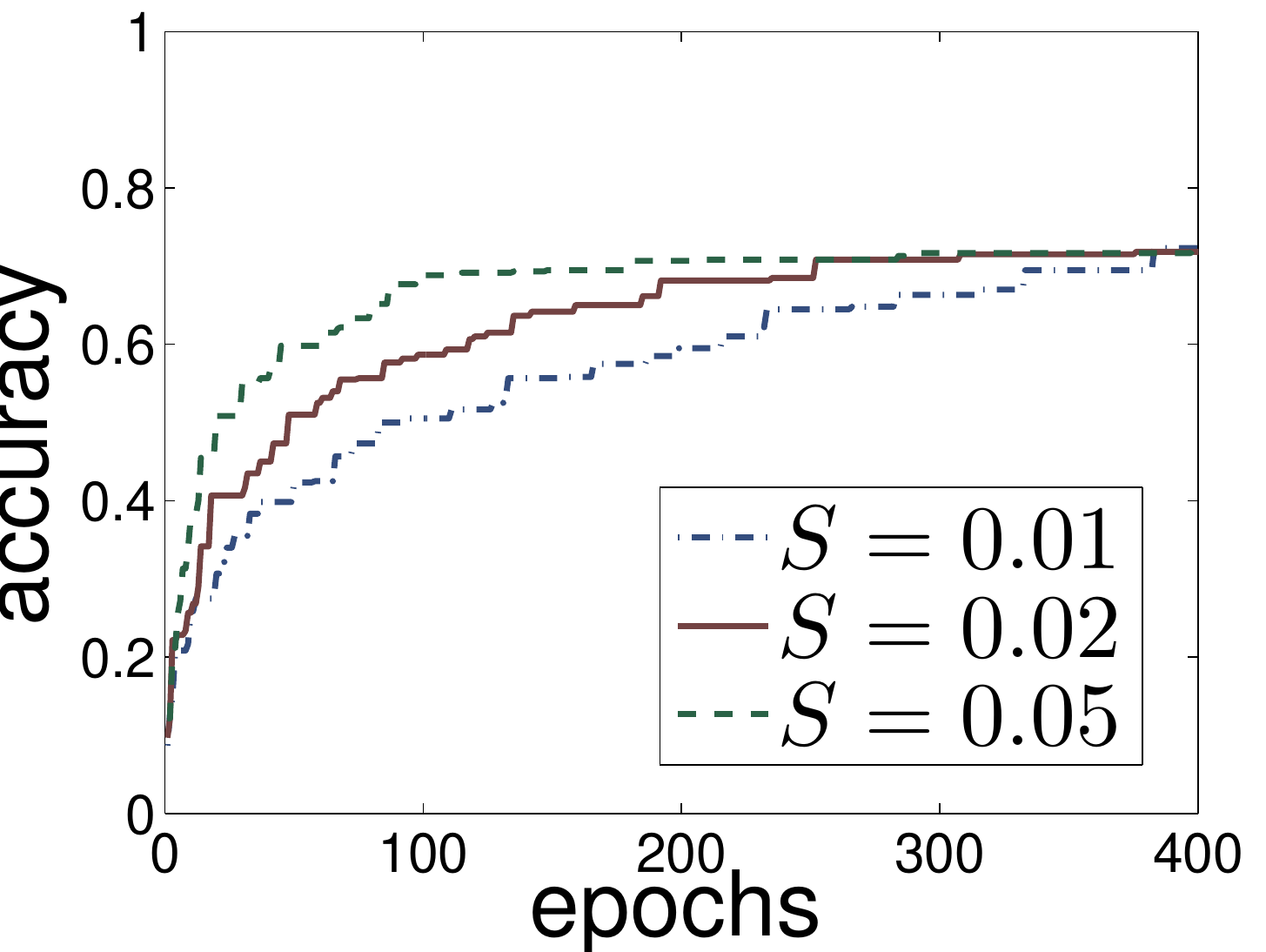}
\end{minipage}
}
\centering
\caption{Experiment results of federated mining.}
\label{lab1}
\end{figure}

\subsection{Accuracy calculation}
After federated mining, the model of each pool is translated and encrypted to calculate the accuracy based on the privacy-preserving model verification mechanism. We conduct the numerical analysis on both computation and communication cost of HE and $GC$ construction affected by data quantity.

We first study the communication cost between the pool manager and the requester, 
changing with data quantity. Assuming that there are four layers with 256, 16, 16, 10 nodes respectively in the model, the communication cost of HE and OT of $GC$ is shown in Fig. \ref{lab_2PC}(a). As we can see, the communication cost of HE is almost linear to the quantity of data. 
This is because the exchanged data are $\left \langle a_{ij}\right \rangle$, $\left \langle z_k+h_k\right \rangle$ and $z_k+h_k$, which are linearly increasing with data quantity.
In addition, the communication cost of OT in $GC$ is also increasing with data quantity since it is related to the number of garbled gates. According to 
Table \ref{tab3}, the number of garbled gates is changing with $I \times \left \lceil log(I) \right \rceil$, where $I$ is exactly the data quantity.

\begin{figure}[ht]
\centering
\subfigure[Communication Cost]{
\begin{minipage}[t]{0.45\linewidth}
\centering
\includegraphics[width=1.0\textwidth]{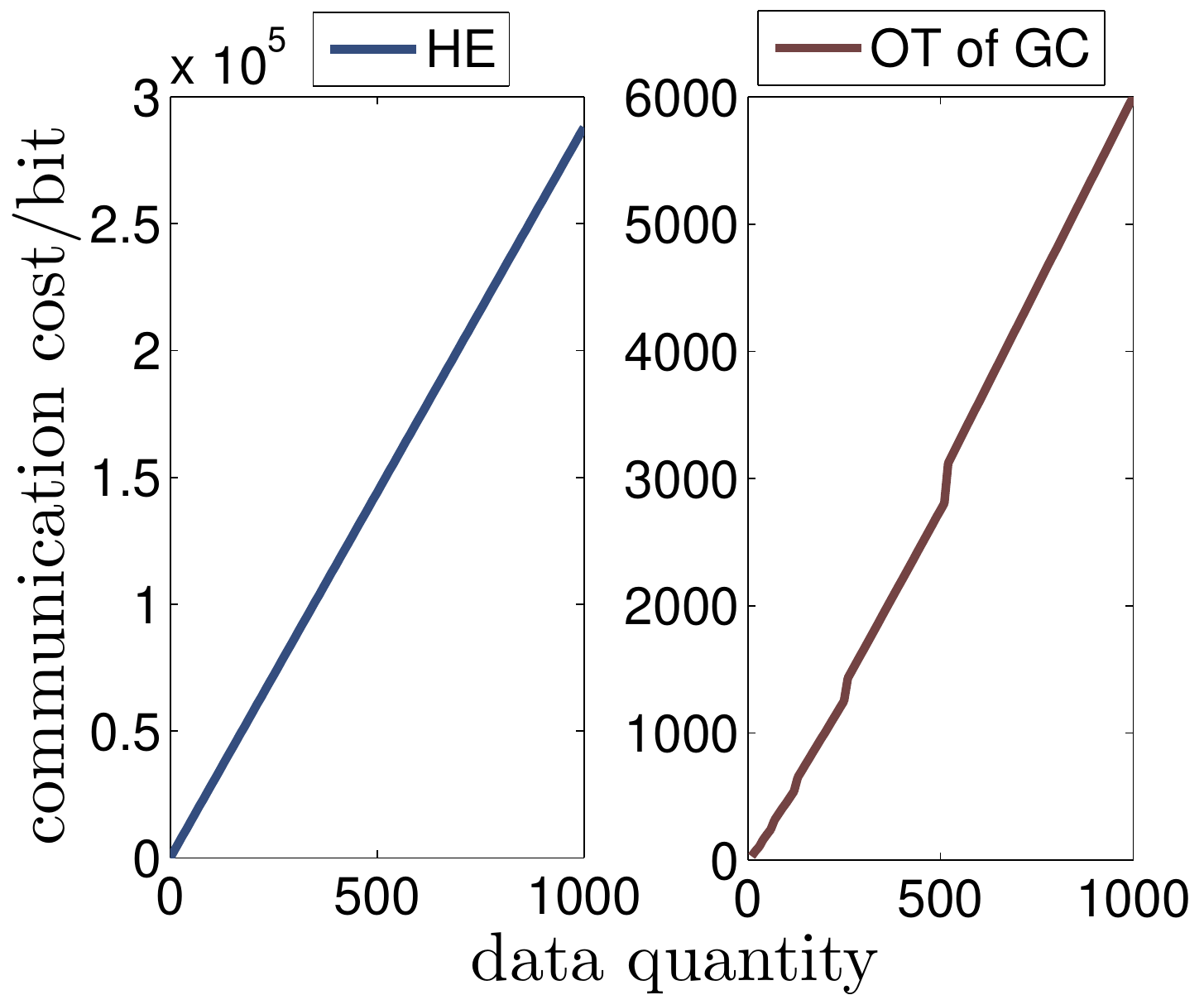}
\end{minipage}
}
\subfigure[Computation Cost]{
\begin{minipage}[t]{0.45\linewidth}
\centering
\includegraphics[width=1.0\textwidth]{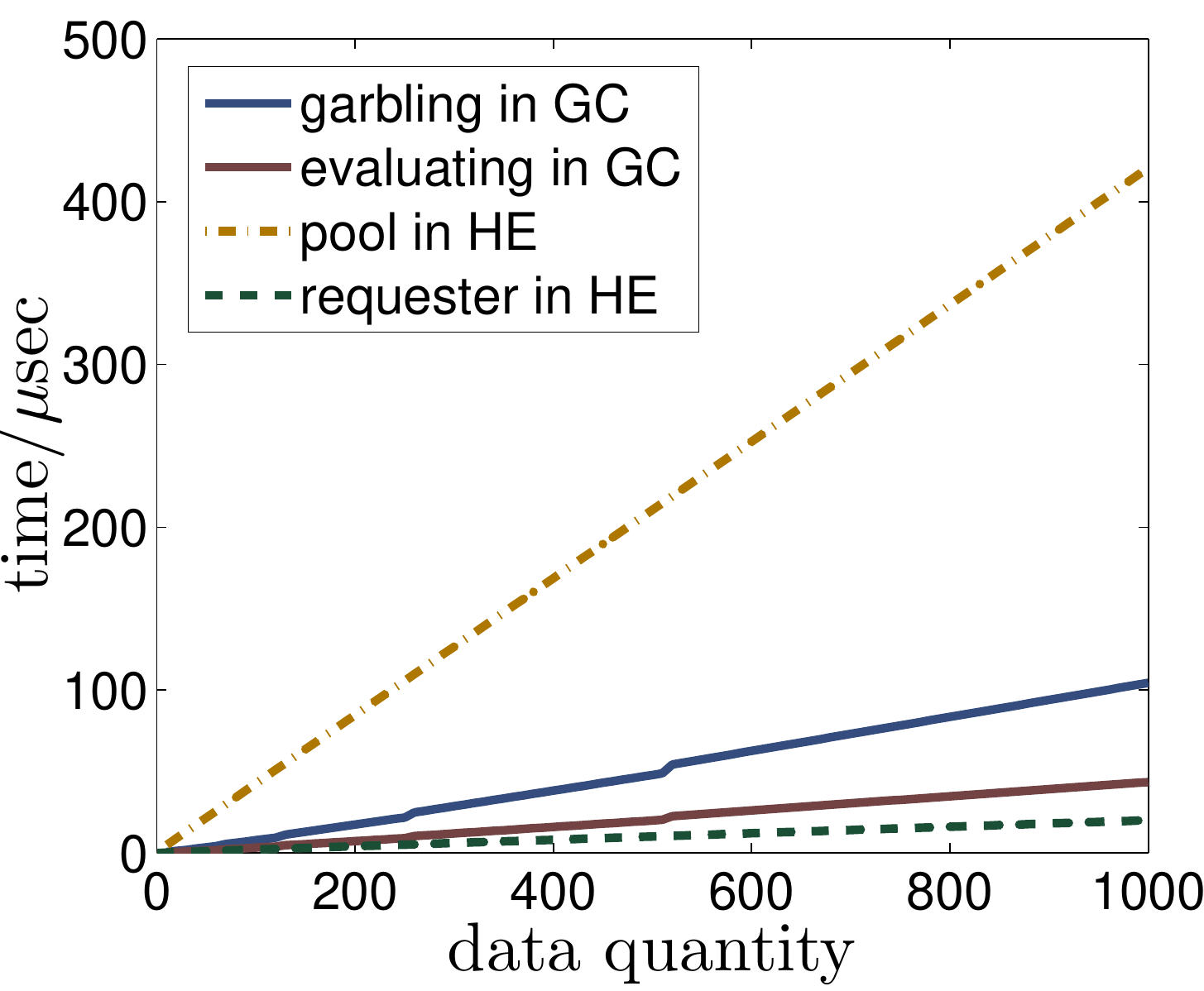}
\end{minipage}
}
\centering
\caption{Experimental results of the model verification mechanism.}
\label{lab_2PC}
\end{figure}

Then we present the computation costs of both HE and GC in Fig. \ref{lab_2PC}(b).
During HE, both the costs of the pool and the requester increase linearly with  data quantity  but
the cost of the pool increases faster than that of the requester. This is
because the requester can complete a lot of calculation offline, such as the encryption of the test data, then the online calculation only refers to the decryption of $\left \langle z_k+h_k\right \rangle$; while for the pool, he has to calculate a lot of intermediate results, such as $z_k$ and $h_k$, which costs him more compared to that of the requester when data quantity increases.
While the time cost of GC consists of two parts, garbling the circuit of the pool and evaluating GC of the requester. We use the best PRM-based garbling scheme of JustGarble system called GaX and our processor runs at 3.20 GHz \cite{Mihir13}. The consumption of GC construction is proportional to the number of non-free gates in the circuit, which is also changing with $I \times \left \lceil log(I) \right \rceil$.

\subsection{Accuracy verification}
When blocks are generated, it comes to the accuracy verification process of full nodes. The time for full nodes to verify accuracy and reach a consensus is very significant. The shorter the time, the better the efficiency and the higher the security.
In our proposed mechanism, the main time spent on verification is accuracy sorting and testing.
The time of sorting is based on  the operating speed and the number of blocks. For the operating speed, we take ASIC in \cite{Abdel17} as an example.  Referring to the average sorting time of ASIC, the time consumption increases with the number of models as shown in Fig. \ref{verification_time}(a). While for the number of received blocks, according to the statistics of Bitcoin \cite{BTC19}, there are 9,364 nodes at present.
Thus, the time of sorting all the accuracy of received blocks in our system will not exceed 60 microseconds for full nodes.

After sorting the accuracy of models in a descending order, full nodes test these models from the first one. If the model with the highest accuracy is verified to be true, there is no need to test other models. Otherwise, the second one will be tested until finding the first model with verified accuracy being the same as that stored in the block header. The quantity of our test data is set as 10,000. The average time for full nodes to test the accuracy of models is between 1.65 seconds and 1.7 seconds as shown in Fig. \ref{verification_time}(b). In addition to the time of sorting, all the time we need is no more than 2 seconds.

\begin{figure}[h]
\centering
\subfigure[Sorting]{
\begin{minipage}[t]{0.45\linewidth}
\centering
\includegraphics[width=1.0\textwidth]{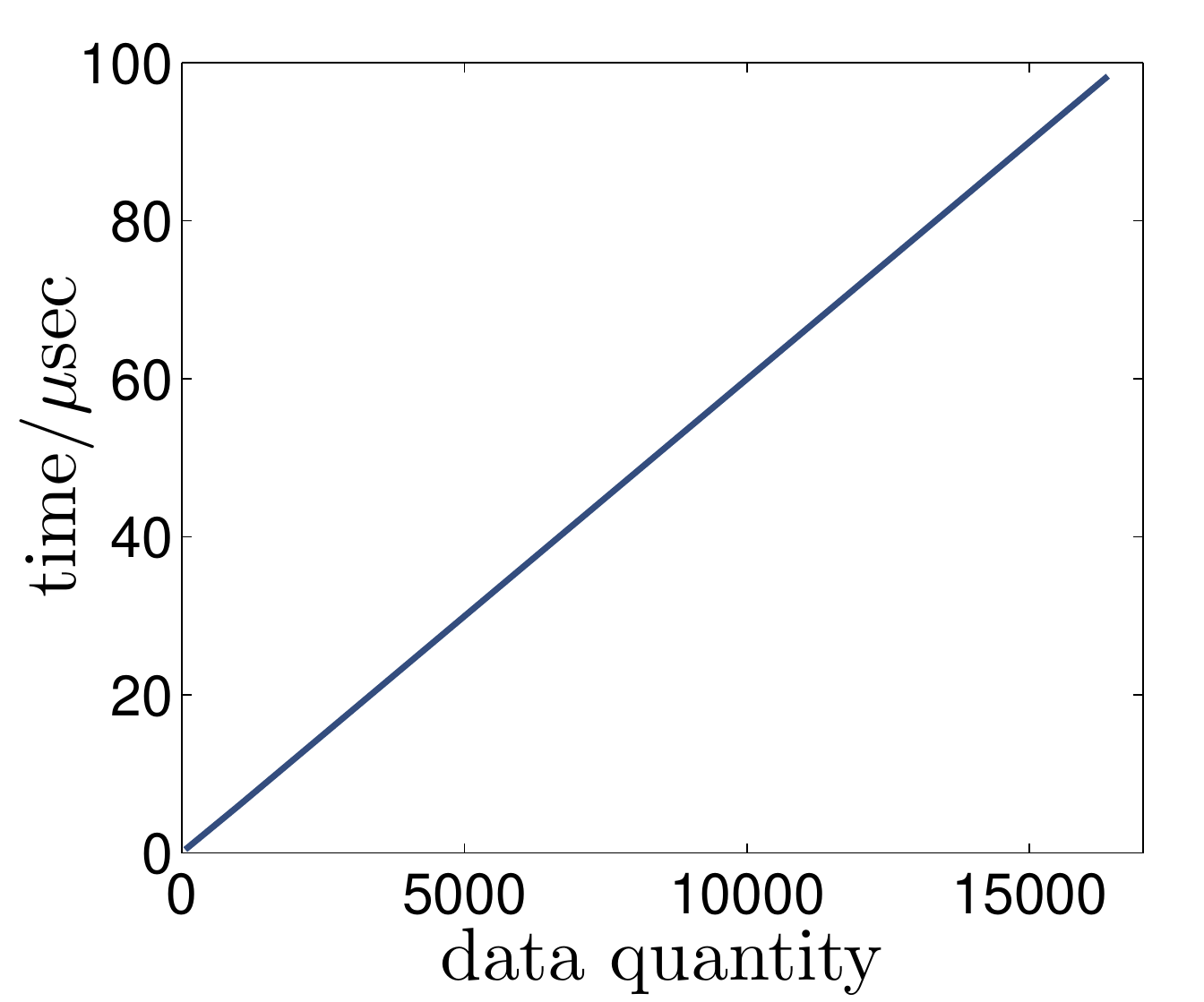}
\end{minipage}
}
\subfigure[Verification]{
\begin{minipage}[t]{0.45\linewidth}
\centering
\includegraphics[width=1.0\textwidth]{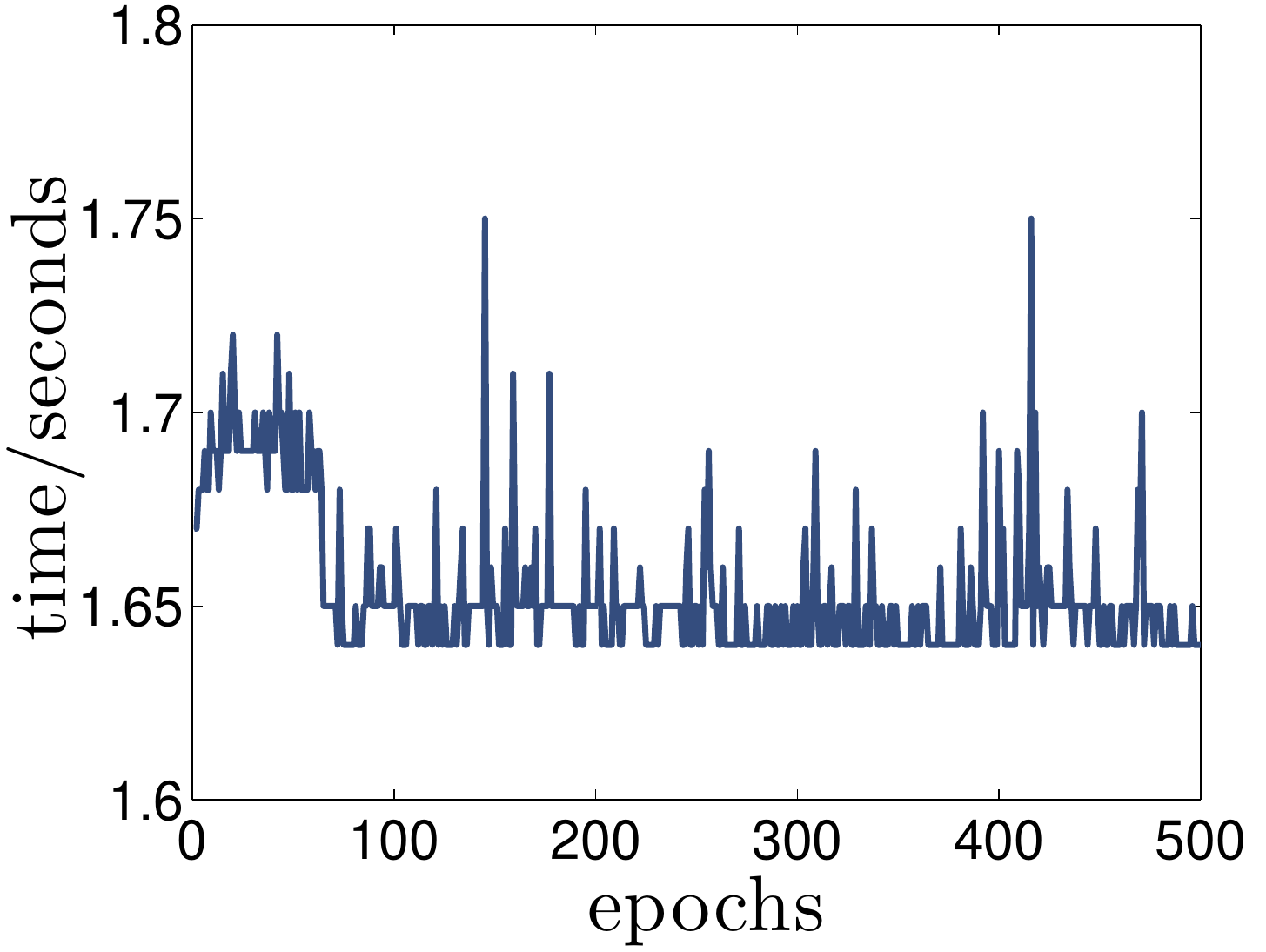}
\end{minipage}
}
\centering
\caption{Time of accuracy verification.}
\label{verification_time}
\end{figure}

\section{Related Work}
\label{sec:related}
Many attempts have been made to find valuable work as a substitute of puzzles in PoW, which should be difficult to solve but easy to verify. Initially, researchers take a small step forward replacing the nonce with some mathematic problems. Primecoin \cite{King13} requires miners to find long prime chains for the proof of work. There are many similar problems \cite{Ball17}, such as orthogonal vectors, all-pairs shortest paths problem and rSUM, which requires to find r numbers to have their sum be zero. However, there is no significant  practical value in solving these mathematical problems at the cost of ridiculously
large amounts of energy.

Next, more diverse and improved attempts are proposed. In PoX (proof of exercise) \cite{Shoker17}, scientific computation matrix-based problems are sent by employers to a third-party platform where miners select tasks to solve. Permacoin \cite{Miller14} proposes PoR (proof of retrievability) to investigate the storage space and memory for a file or file fragment, where the mining is not associated with computation but storage resources. PieceWork \cite{Daian17} reuses the wasted work for additional goals such as spam prevention and DoS mitigation by outsourcing.

In recent, the combination of deep learning and Blockchain has appeared. A substitution of PoW named PoDL (proof of deep learning) \cite{Chenli19} first uses deep learning for Blockchain maintenance instead of useless hash calculation, which only needs to add some new components to the block header and thus can be applied to the current cryptocurrency systems. However, the user needs to provide a complete training dataset for model training and test datasets for verification to all miners, which sacrifices data privacy and may frustrate the data provider to apply Blockchain. Besides, image segmentation is employed in \cite{Li19}, which defines the segmentation model as proof. Furthermore, Coin.AI \cite{Baldominos19} proposes a proof-of-storage scheme to reward users for providing storage for deep learning models.

\section{Conclusion}
\label{sec:conclusion}
In this paper, we propose a novel energy-recycling consensus algorithm named PoFL, where the cryptographic puzzles in PoW is replaced with federated learning tasks. To realize PoFL, a general  framework is introduced and  a new PoFL block structure is designed  for supporting block verification. To guarantee the privacy of  training data, a reverse game-based data trading mechanism is proposed,  which takes advantage
of market power to make a rational pool maximize his utility
only when he trains the model without any data leakage, thus further motivating pools to behave well. In addition, a privacy-preserving model verification mechanism is designed to verify the accuracy of a trained model while  preserving  the privacy of the task requester's test data as well as avoiding the pool's submitted model to be plagiarized by others, which employs homomorphic encryption and secure two-party computation in label prediction and comparison, respectively. Extensive simulations based on synthetic and real-world data demonstrate the effectiveness and efficiency of PoFL.

\bibliography{reference}
\bibliographystyle{IEEEtran}

\end{document}